\documentclass[aps,twocolumn,preprintnumbers,a4paper]{revtex4}

\usepackage{amsmath, amssymb, amsfonts}
\usepackage{amsthm}
\usepackage{physics}  
\usepackage{bm}       

\usepackage{graphicx}     
\usepackage{multirow}
\usepackage{placeins}     

\usepackage{caption}
\usepackage{subcaption}

\usepackage{tikz}
\usepackage{dcolumn}      
\usepackage{xcolor}
\usepackage[T1]{fontenc}

\usepackage{fancyhdr}
\usepackage{parskip}
\usepackage{hyperref}

\newtheorem{theorem}{Result}[section]

\newtheorem{corollary}[theorem]{Corollary}

\begin{document}

\title{Hierarchy of Qubit Dynamical Maps in the Presence of Symmetry and Coherence}

\author{Unnati Akhouri}
\email{uja5020@psu.edu}
\affiliation {\it Institute for Gravitation and the Cosmos, The Pennsylvania State University, University Park, PA 16802, USA and\\
Department of Physics, The Pennsylvania State University, University Park, PA 16802, USA.}

\date{\today}

\begin{abstract}  
We investigate the relationship between symmetries and thermodynamic transformations by analyzing how global energy conservation and coherence resources affect the local dynamics of subsystems. We prove that U(1) conservation fundamentally constrains quantum thermodynamic operations through charge conservation of Pauli strings. Our no-go theorem shows that U(1) dynamics cannot generate local coherence from diagonal thermal states, restricting thermal operations to phase-covariant maps. Breaking this hierarchy requires environmental coherence with odd charge parity that couples unequal energy states. We establish the minimal resource requirements through a two-qubit construction that achieves Gibbs-preserving transformations beyond thermal operations. We demonstrate measurable thermodynamic advantages in work extraction and state distinguishability, revealing the fundamental role of quantum coherence in enhancing thermodynamic performance.
\end{abstract}

\maketitle 

\section{Introduction}
The dynamics of a subsystem interacting with a larger environment is shaped by its interactions and the environment's state  \cite{lidar2020lecturenotestheoryopen}. Constraints imposed on these interactions or states inherently restrict the subsystem's evolution \cite{RevModPhys.91.025001}.  To fully characterize the evolution of the subsystem, it is crucial to analyze the available resources and the operations allowed within the given set of constraints \cite{Kolchinsky_2025}. Quantum resource theories provide a robust framework for identifying the possible transformations under constraints and the specific resources needed to achieve transformations that would otherwise be unattainable \cite{RevModPhys.91.025001,Lostaglio_2019}.\\
A central question that arises is what determines which states are accessible and what resources are required to unlock a broader range of states \cite{Brand_o_2015}. In the thermodynamic context, this includes frameworks that quantify coherence under symmetry constraints \cite{Brand_o_2015,spekkens2007evidence,marvian2016coherence} and capture the role of catalytic coherence and thermo-majorization in state transformations \cite{Brand_o_2015}.\\
In quantum thermodynamics, thermal operations—arising from energy-conserving interactions with thermal reservoirs—have been extensively studied as the paradigmatic class of operations. These operations conserve the total system-environment energy and do not require external work. For a system whose energy basis aligns with the computational basis, these operations naturally emerge as phase-covariant dynamical maps that describe the effective evolution of a subsystem within the full system-environment composite \cite{Faist_2015,Filippov2020,ROGA2010311,Horodecki_2013}. Phase-covariance suppresses creation of coherence between energy eigenstates thereby limiting the accessible states. Any state possessing coherence or athermality is considered a resource within this theory since it allows accessing a broader class of states \cite{Horodecki_2013}.\\
However, thermal operations are insufficient to describe all thermodynamically relevant processes. Using these operations as a baseline, we can ask what resources are necessary to go beyond thermal operations? In this work, we show that U(1) symmetry corresponding to the energy-conservation, with the interplay of states, creates a rigid hierarchical structure in the space of allowable transformations for qubit systems. \\
 If the environment lacks certain classes of  coherences, U(1)-symmetric evolution restricts dynamics to phase-covariant thermal operations, establishing that special classes of coherences are needed to go beyond thermal operations. In particular, 
 we show that coherence terms mixing states with unequal energy break phase-covariance.  We clarify how distinct classes of coherences affect different components of dynamical map describing the effective evolution of the subsystem. We prove our results through charge conservation of Pauli strings and connect it to the emergence of thermodynamic hierarchies. Unlike previous phenomenological approaches, our symmetry-based analysis reveals the fundamental mathematical structure underlying these limitations.\\
We use our methodology to demonstrate the generation of Gibbs-preserving (GP) maps that have been previously studied as a class distinct from thermal operations. These dynamics can generate coherences between energy eigenstates making them particularly profound in the quantum regime, where coherence serves as a genuine resource that can enhance thermodynamic performance 
\cite{shiraishi2024quantumthermodynamicscoherencecovariant, Faist_2015,ende2023exploringlimitsopenquantum,sagawa2020entropydivergencemajorizationclassical,PhysRevLett.134.170201,Shiraishi_2025} While previous work has established that Gibbs-preserving maps outperform thermal operations \cite{sagawa2020entropydivergencemajorizationclassical, Faist_2015}, the minimal resources required to generate these effective dynamics under global unitary dynamics have remained unclear. Our work aims to fill this gap by providing explicit no-go theorems and constructive protocols to generate GP dynamics. \\
To this end, we present systems of two, three qubits evolving with $\sigma_X\sigma_X+\sigma_Y\sigma_Y$ interaction. We adopt the formalism of completely-positive, trace-preserving (CPTP) maps for subsystems of qubits within the full system. The maps provide a clear description of how a system evolves under the influence of its environment, accounting for dissipation and decoherence. We demonstrate how specific correlations can be harnessed to engineer different classes of subsystem dynamics including Gibbs-preserving dynamics. \\
A comprehensive understanding of the limitations and capabilities of global energy preserving operations is crucial for developing quantum devices and protocols, particularly within the noisy intermediate-scale quantum (NISQ) era  \cite{roadmap}. This is especially relevant for thermodynamics at the nano-scale where quantum coherences can enhance the work extraction \cite{Horodecki_2013}. Our minimal three-qubit construction provides an experimentally realizable protocol for going beyond thermal operations, with explicit parameter relationships that can guide implementation on current quantum platforms. We quantify thermodynamic advantages through enhanced work extraction capabilities and increased state distinguishability.\\
In Section \ref{Sec:theory_strings}, we review the theoretical framework of Pauli strings and unitaries used throughout the paper. In Section \ref{Sec:QubitDynMap}, we review the framework of qubit dynamical maps in affine form. In Section \ref{Sec:Hierarchy}, we categorize the hierarchy of thermodynamics operations and in Section \ref{Sec:min_model} we introduce the minimal model used in the paper. 
Finally, in Section \ref{Sec:Results} we present our results for the classes of maps and demonstrate utility of going to classes of maps beyond thermal operations. 
\begin{figure}[htbp]
    \centering
    \begin{tikzpicture}[scale=1]
        \begin{scope}[xshift=-2.5cm]
            \draw[thick] (-2,-1.5) rectangle (1.5,1);
            \node at (1.5,1.2) {$\Phi$};

            \draw[thick, gray,fill=gray,fill opacity=0.2] (-0.1,-0.2) ellipse (1.4cm and 1.cm);
            \node at (0.9,0.7) {$\Phi_E$};
            
            \draw[thick, cyan,fill=cyan,fill opacity=0.2] (-0.8,-0.2) ellipse (1.1cm and 0.8cm);
            \node at (0.3,0.52) {$\Phi_{GP}$};

            \draw[thick, cyan,fill=cyan,fill opacity=01] (-0.7,0.45) circle (.05cm and 0.05cm);
            \node at (0.3,0.52) {$\Phi_{GP}$};

            \draw[thick, magenta,fill=magenta,fill opacity=01] (-0.4,-0.4) circle (.05cm and 0.05cm);
            \node at (0.3,0.52) {$\Phi_{GP}$};

            \draw[thick, orange,fill=orange,fill opacity=01] (-1.7,0.1) circle (.05cm and 0.05cm);
            \node at (0.3,0.52) {$\Phi_{GP}$};

            \draw[thick, black,fill=black,fill opacity=01] (0.5,0.1) circle (.05cm and 0.05cm);
            \node at (0.3,0.52) {$\Phi_{GP}$};
            
            \draw[thick, magenta,fill=magenta,fill opacity=0.2] (-0.7,-0.2) ellipse (0.8cm and 0.5cm);
            \node at (-0.7,0) {$\Phi_{PC}$};
        \end{scope}
        
        \begin{scope}[xshift=2.5cm]
            \draw[thick] (-2,-1.5) rectangle (1.5,1);
            \node at (1.5,1.2) {$\Phi$};
            
            \draw[thick, cyan,fill=cyan,fill opacity=0.2] (-0.3,-0.3) ellipse (1cm and 0.8cm);
            \draw[thick, magenta,fill=magenta,fill opacity=0.2] (-0.3,-0.3) ellipse (1cm and 0.8cm);
            \node at (0,0.7) {$\Phi_{GP}$};
            \node at (0,-1.3) {$\Phi_{PC}$};
        \end{scope}
    \end{tikzpicture}
\caption{Classically, Gibbs-preserving and phase-covariant maps are identical. However, in quantum mechanics due to the presence of coherences the phase-covariant maps arising from thermal operations are different from the Gibbs-preserving maps and dynamical maps follow $\Phi_{PC} ^{\textcolor{magenta}{\bullet}}\subset \Phi_{GP(E)}^{\textcolor{cyan}{\bullet}} \subset \Phi_E ^{\textcolor{black}{\bullet}}\subset\Phi $. This means Phase-covariant maps are a proper subset of Gibbs-preserving maps however, there are Gibbs-preserving maps that are not phase-covariant. In this paper we will characterize how to go from covariant operations to a phase-covariant map and what extra resources are needed to get to a Gibbs-preserving map. }
\label{fig:PC_GP_map_sets}
\end{figure}
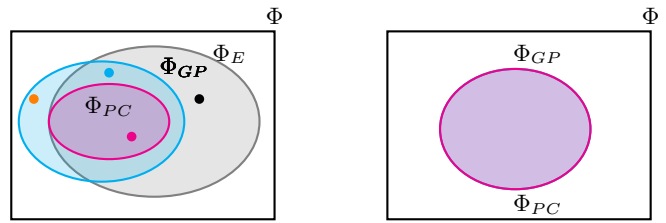
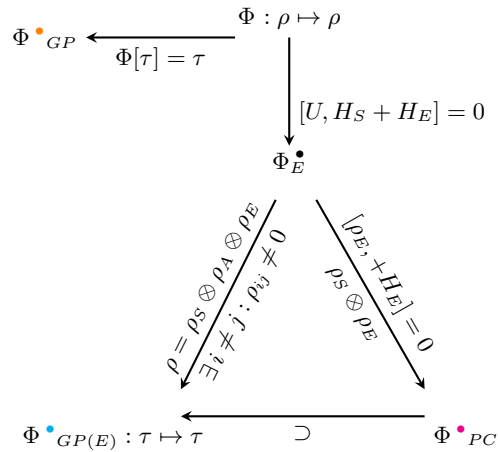
\begin{figure}[htbp]
    \centering
\begin{tikzpicture}[scale=1.8]
    \coordinate (O) at (1, 2.8);
    \coordinate (GP) at (-0.5, 2.8);
    \coordinate (A) at (-0.3, 0);
    \coordinate (B) at (2, 0);
    \coordinate (C) at (1, 1.732);
    
    \draw[->, thick, >=stealth] (0.6, 2.8)--(-0.5, 2.8) node[midway, below] {$\Phi[\tau] = \tau$};
    \draw[->, thick, >=stealth] (1, 2.8)--(1, 2) node[midway, below right] {$[U,H_S+H_E]=0$};
    \draw[->, thick,>=stealth] (2, 0) -- (0.2, 0) node[midway, below, sloped] {$\supset$};
    \draw[->, thick,>=stealth] (1.2, 1.6) -- (2, 0.2) 
    node[midway, above, sloped] {$[\rho_E,+H_E]=0$}
    node[midway, below, sloped] {$\rho_S\otimes \rho_E$};
    \draw[->, thick,>=stealth] (0.9, 1.6) -- (0.2,0.2) node[midway, above, sloped] {$\rho=\rho_S\otimes\rho_A\otimes\rho_E$}node[midway, below, sloped] {$\exists \hspace{0.2em}    i\neq j: \rho_{ij} \neq 0 \hspace{0.2em}$};
    \node at (O) [above] {$\Phi_{}: \rho \mapsto \rho$};
    \node at (A) [below] {${\Phi^{\hspace{0.2em}\textcolor{cyan}{\bullet}}}_{GP(E)}:\tau\mapsto \tau$};
    \node at (GP) [left] {${\Phi^{\hspace{0.2em}\textcolor{orange}{\bullet}}}_{GP}$};
    \node at (B) [below right] {${\Phi^{\hspace{0.2em}\textcolor{magenta}{\bullet}}}_{PC}$};
    \node at (C) [above] {${\Phi^{  \hspace{0.2em}\textcolor{black}{\bullet}}_{E}}$};
\end{tikzpicture}
\caption{A schematic for the conditions needed to get maps. We start with $\Phi[\circ]$ which just defines the entire set of CPTP and Hermiticity preserving dynamical maps. Applying energy conservation on the composite gives the subset of what we call E-preserving maps $\Phi_E ^{\textcolor{black}{\bullet}}$, whereas, demanding preservation of Gibbs-state gives the general Gibbs-preserving map $\Phi_{GP}^{\textcolor{orange}{\bullet}}$. Demanding conditions on the environment state and initial correlations gives $\Phi_{PC} ^{\textcolor{magenta}{\bullet}}$ and the addition of resource qubit with coherence results in covariant Gibbs-preserving maps $\Phi_{GP(E)}^{\textcolor{cyan}{\bullet}}$. }
\end{figure}

\section{Theoretical Framework}
\label{Sec:theory_strings}
\subsection{Pauli strings and ladders}
We consider $n$-qubit systems, with the Hilbert space given by $\mathcal{H} = (\mathbb{C}^2)^{\otimes n}$. A complete orthonormal basis for operators acting on this space is given by the $4^n$ Pauli strings, formed by tensor products of single-qubit Pauli operators ($I, \sigma_X, \sigma_Y, \sigma_Z$):
$$O_k = P_{1k}\otimes P_{2k}\otimes\ldots \otimes P_{nk}$$
where $P_{jk} \in {I, \sigma_X, \sigma_Y, \sigma_Z}$ for each qubit $j \in {1, \ldots, n}$, and $k \in {1, \ldots, 4^n}$ indexes the $4^n$ distinct Pauli strings. Any operator $M$ in this Hilbert space, including density matrices and unitary operators, can be expressed as a linear combination of these basis operators:
$$M = \sum_k \alpha_k O_k$$
For what follows, we introduce the terminology of the charge of a Pauli string, $O_k$. The charge $C(O_k)$ tracks the parity of coherence-generating operators in a Pauli string and is defined as 
\begin{equation}
    C(O_k) = (-1)^{n_x+n_y}.
\end{equation}
Here, $n_x$ and $n_y$ count the number of $\sigma_x$ and $\sigma_y$ operators respectively. The charge is +1 when $n_x+n_y$ is even and -1 when the sum is odd. We use this to prove in the Appendix \ref{App3:Paulicharge} that U(1) symmetric evolution preserves charge, meaning operations starting from diagonal states (all charge +1 terms) cannot generate local coherence (charge -1 terms).
The Pauli string representation allows us to track how different components of density matrix transform under evolution, particularly the coherence versus population terms.

To analyze the constraints imposed by U(1) symmetry, it is often convenient to employ a basis related to ladder operators for each qubit, instead of the conventional Pauli basis. For a single qubit $j$, these ladder operators are defined as: $\sigma^+_j = |1\rangle_j\langle 0|$ and
$\sigma^-_j = |0\rangle_j\langle 1|$
These operators are directly related to the Pauli operators by $\sigma_X = \frac{1}{2}(\sigma^+ + \sigma^-)$ and $\sigma_Y = \frac{1}{2i}(\sigma^+ - \sigma^-)$. This basis naturally reflects the U(1) symmetry structure, particularly when considering a conserved quantity like the total charge, $\sum_i \sigma_{zi}$. Such Hamiltonians will only contain terms involving products of $\sigma^+_j$ and $\sigma^-_j$ operators where the total number of raising operators equals the total number of lowering operators, thus preserving the overall excitation number.

\subsection{U(1) Symmetry and Energy Conservation}
 For a composite system with total Hamiltonian $H = H_S + H_E + H_{SE}$, U(1) symmetry is established when the  interaction Hamiltonian commutes with the free Hamiltonian. For qubits, we define the U(1) generator to be the particle number operator:
$$G = \sigma_Z\otimes I + I \otimes \sigma_Z$$
The commutation relation $[H, G] = 0$ ensures the conservation of the total excitation/particle number (or magnetization along the z-axis) and fundamentally constrains how operators transform under the dynamics.\\
The unitary operator associated with this dynamics also preserves the symmetry and possesses a block-diagonal structure in the basis of conserved excitation number. This means that $U$ can be decomposed into independent blocks, each acting only on subspaces corresponding to a fixed number of excitations (qubits in the $|1\rangle$ state):
$$U=\mathcal{U}_0 \oplus \mathcal{U}_1 \oplus \mathcal{U}_2 \ldots \oplus \mathcal{U}_N$$
Here, $\mathcal{U}_m$ is an independent unitary matrix acting within the subspace spanned by states with exactly $m$ qubits in the $|1\rangle$ state. The dimension of each $\mathcal{U}_m$ is given by $\binom{N}{m}$, the number of ways to choose $m$ qubits to be in the $|1\rangle$ state out of $N$. This block-diagonal structure implies that U(1)-symmetric unitaries cannot change the total excitation number of a state, restricting the types of transformations they can implement. In the language of Pauli strings, this translates to charge conservation property: operators containing an even number of $\sigma_x$ and $\sigma_y$ terms can only evolve into operators with an even number of such terms, while operators with an odd number evolve only into odd-charge operators.
\section{Qubit Dynamical Maps}
\label{Sec:QubitDynMap}
\subsection{Framework for Open Quantum Evolution}

We consider the evolution of a quantum system A interacting with an environment B, initially in a joint state $\rho_{AB}$. If the total bipartite system undergoes a unitary evolution $U$, the effective transformation of system A, after tracing out environment B, is described by a quantum dynamical map $\mathcal{E}$. This process can be formally expressed as:
\begin{equation}
\rho_A \mapsto \mathcal{E}(\rho_A) = \text{Tr}_B \left[ U (\rho_A \otimes \rho_B) U^\dagger \right]
\end{equation}
Here, $\rho_B$ is the initial state of the environment. Since this operator arises from a joint unitary evolution on uncorrelated states, the evolution $\mathcal{E}$ is guaranteed to be a completely positive trace-preserving (CPTP) map \cite{Jagadish_2018}. A quantum channel $\mathcal{E}$ describes the finite-time open dynamics of a quantum system, mapping an initial quantum state to its final state after interaction with the environment. Any dynamical map $\mathcal{E}$ on the states of a quantum system must be linear and preserve Hermiticity, positivity, and trace—properties that collectively ensure valid density matrix transformations.

\subsection{Bloch Vector Representation and Affine Form}

For a single qubit, the dynamical map acts on the space of $2\times2$ complex Hermitian matrices. The Pauli matrices along with the $2\times2$ identity matrix $(I, \sigma_X, \sigma_Y, \sigma_Z)$ form a natural and complete basis for this space. A general qubit state (density matrix) $\rho$ can be represented in terms of its Bloch vector $\vec{a} = (a_1,a_2,a_3)^T$ as:
\begin{equation}
\rho = \frac{1}{2} (I + \vec{a} \cdot \vec{\sigma}) = \frac{1}{2} \begin{pmatrix} 1+a_3 &a_1 - ia_2 \\a_1 + ia_2 & 1-a_3 \end{pmatrix}
\end{equation}
where $\vec{\sigma} = (\sigma_X, \sigma_Y, \sigma_Z)$ are the Pauli matrices. The components of the Bloch vector are constrained by $a_1^2 +a_2^2 +a_3^2 \leq 1$. This allows for a geometric visualization of any qubit state as a point within the three-dimensional unit Bloch sphere: pure states reside on the surface, while mixed states lie within its interior.

Under a quantum dynamical map $\mathcal{E}$, the initial Bloch vector $\vec{a}$ transforms to a new Bloch vector $\vec{a}'$. This transformation is linear and affine, and can be generally expressed as:
\begin{equation}
\vec{a}' = \vec{\tau} + T\vec{a}
\end{equation}
Here, $\vec{\tau}$ is the three-dimensional shift vector representing a translation or shift of the Bloch vector, and $T$ is a $3\times3$ real matrix that scales, rotates, and mixes the components of the initial Bloch vector.\\
The shift-vector captures the non-unitality and corresponds to the evolution of Pauli strings with support in the environment to terms that have support on the system qubit. Consider the system-environment decomposed as a seperable and a correlated matrix such that $\rho_{SE} = \rho_S \otimes \rho_E + \chi_{SE}$. It's unitary evolution under $U$ is:
\begin{align*}
    U\rho_{SE}U^\dagger &= U\rho_S\otimes \rho_EU^\dagger + U\chi_{SE}U^\dagger\\
&=U(\mathbf{1}+\vec{a}.\vec{\sigma})\otimes\rho_E^\dagger+ U\chi_{SE}U^\dagger\\
&=U(\mathbf{1}\otimes\rho_E + \chi_{SE})U^{\dagger}+U(\vec{a}.\vec{\sigma}\otimes \rho_E) U^\dagger
\end{align*}
The shift vector $\vec{\tau}$ captures how the environment and initial correlations contribute to the system evolution, while the transformation matrix $T$ describes how the initial system Bloch vector components are mixed and scaled. 

\subsection{Matrix Representation}

The affine transformation can be compactly represented by an augmented $4\times4$ matrix, referred to as the affine form of the dynamical map \cite{Jagadish_2018,ruskai2001analysiscompletelypositivetracepreservingmaps,Filippov2020}. This map $\Phi$ acts on an augmented vector $(1, \vec{a})^T$:
\begin{equation}
\begin{pmatrix} 1 \\ \vec{a}' \end{pmatrix} = \begin{pmatrix} 1 & \mathbf{0}^T \\ \vec{\tau} & T \end{pmatrix} \begin{pmatrix} 1 \\ \vec{a} \end{pmatrix}
\end{equation}
where $\mathbf{0}^T$ is a $1\times3$ zero vector. This compact form offers several utilities for characterizing qubit dynamics:

\begin{itemize}
\item \textbf{Trace preservation}: The first row of $\Phi$ inherently ensures that the map is trace-preserving, since the trace of $\rho$ is related to the first component of the augmented Bloch vector, which remains invariant.

\item \textbf{Unitality characterization}: The shift vector $\vec{\tau}$ directly determines whether the map is unital. A map is unital if it leaves the maximally mixed state ($\rho = I/2$, corresponding to $\vec{a} = \mathbf{0}$) invariant, which occurs if and only if $\vec{\tau} = \mathbf{0}$.

\item \textbf{Unitary constraints}: For a unitary process, $T$ is a rotation matrix satisfying $TT^T = I$. For general open quantum evolution, $T$ is a more general matrix reflecting dissipation and decoherence.

\item \textbf{CPTP conditions}: The constraints for a map to be CPTP translate into specific conditions on the matrix $T$ and vector $\vec{\tau}$.
\end{itemize}

In this work, we derive these affine maps for qubit subsystems by considering unitary evolution of larger composite systems followed by tracing out environmental degrees of freedom. This approach provides a non-perturbative way to describe the effective open-system dynamics, allowing us to connect the properties of the overall unitary interaction to the characteristics of the resulting qubit channel.

\section{Hierarchy of Thermodynamic Operations}
\label{Sec:Hierarchy}
The constraints imposed by U(1) symmetry create a natural hierarchy in the space of thermodynamic operations. We characterize this by analyzing how different initial conditions and resource allocations lead to distinct classes of dynamics. 
\subsection{Globally-covariant Dynamics}
 We consider qubit systems where the total, corresponding to the total particle number $\sum_i \sigma_{zi}$, is conserved under unitary evolutions. For a bipartite system AB, the effective transformation of system A, after tracing out environment B, is described by a quantum dynamical map $\Phi_E$. This process can be formally expressed as:
\begin{equation}
\rho_A \mapsto \Phi_E[\rho_A] = \text{Tr}_B \left[ U (\rho_{AB}) U^\dagger \right]
\end{equation}
where $U$ has U(1) symmetry. In the affine form, such dynamics results in a dynamical map with all the entries in rotaiton matrix and shift vector as non-zero. We call these maps globally energy preserving maps, $\Phi_E[\cdot]$. In Section \ref{Sec:Results}, we will compare the action of this map on the qubits.
\subsection{Phase-Covariant Qubit Dynamics}

Phase-covariant maps $\Phi_{\text{PC}}$ are defined by their commutation with the action of a unitary $U(\alpha) = e^{-i \sigma_Z\alpha}$, such that
\begin{equation}
\Phi_{\text{PC}}^{\textcolor{magenta}{\bullet}} \cdot (U(\alpha)\rho_S U(\alpha)^\dagger) = U(\alpha) \Phi_{\text{PC}}^{\textcolor{magenta}{\bullet}}(\rho_S) U(\alpha)^\dagger
\end{equation}
They arise in systems with time translational symmetry and are therefore important for simplifying the analysis of open quantum systems under the influence of noise which preserves certain symmetries \cite{PhysRevX.5.021001}. At the level of qubits where, these maps naturally arise as phase-covariant maps and have been extensively studied for noise and homogenization models \cite{Filippov2020,ROGA2010311,Horodecki_2013,Faist_2015,ziman2001quantumhomogenization}.
In particular, it has been shown that for an initially separable system $\rho = \rho_S \otimes \rho_E$, phase-covariant maps arise when
\begin{align}
[H_I, H_S + H_E] &= 0 \\
[H_E, \rho_E] &= 0
\end{align}
which corresponds to situations in which the state of the environment is thermal and the interaction conserves energy \cite{PhysRevResearch.4.043075}.

The action of a phase-covariant map on the qubit Bloch vector is be given by
\begin{equation}
[a_1,a_2,a_3] \mapsto [\lambda a_1, \lambda a_2, \lambda_za_3 + \tau_z]
\end{equation}
where $\lambda$, $\lambda_z$, and $\tau_z$ are parameters of the map. The action of such a map on the Bloch sphere contracts the sphere into an ellipsoid shifted along the z-axis \cite{Filippov2020}.\footnote{The map is completely positive only when $|\lambda_z| + |\tau_z| \leq 1$ and $4\lambda^2 + \tau_z^2 \leq [1 + \lambda_z]^2$.}

Therefore, the fixed state of the phase-covariant map is given by
\begin{equation}
[0, 0, a_{\text{PC}}^*] = \left[0, 0, \frac{\tau_z}{1 - \lambda_z}\right]
\end{equation}
such that $\Phi_{\text{PC}}^{\textcolor{magenta}{\bullet}}(\rho^*)\mapsto \rho^*$ where $\rho^* = \frac{1}{2}(\mathbf{1} + a_{\text{PC}}^* \sigma_Z)$.

When the energy basis of the qubit aligns with the Z-phase basis, phase-covariant maps and thermal operations overlap. Thermal operations $\Phi_{\text{th}}$ are defined by
\begin{align}
\Phi_{\text{th}}(\rho) &= \Phi_{\text{PC}}^{\textcolor{magenta}{\bullet}}(\rho)=\text{Tr}_B[U(\rho \otimes \tau_B)U^\dagger] \\
\text{where } \tau_B &= \frac{e^{-\beta H_B}}{Z_B}, \quad [U, H_S + H_B] = 0
\end{align}
The above conditions—uncorrelated initial state, Gibbs state for the environment, and energy-preserving interaction—exactly align with the conditions for a phase-covariant map to occur. In this work, we consider the qubit energy basis to be the same as the z-basis. Consequently, thermal operations are always phase-covariant (also called time-covariant owing to time translational symmetry): $\Phi_{\text{th}}(e^{-iH_S t}\rho e^{iH_S t}) = e^{-iH_S t}\Phi_{\text{th}}(\rho)e^{iH_S t}$. By construction, they preserve the Gibbs state: $\Phi_{\text{th}}(\tau_S) = \tau_S$ where $\tau_S = e^{-\beta H_S}/Z_S$, and therefore are a subset of Gibbs-preserving maps. They fall in the class of dynamics that can never create coherence between energy eigenstates. As we will see, this is not true for a general $\Phi_{\text{GP}}$.

\subsection{Gibbs-Preserving Qubit Dynamics}

For a qubit Hamiltonian $H = \frac{\omega}{2}\sigma_z$, the Gibbs state at inverse temperature $\beta = 1/T$ is \footnote{We use units where $\hbar = k_B = 1$.}
\begin{equation}
\rho_G = \frac{e^{-\beta H}}{Z} = \frac{1}{2}(I + r_G \sigma_z)
\end{equation}
where $r_G$ is related to the thermal population in the excited state. A completely positive trace-preserving map $\Phi$ is Gibbs-preserving if
\begin{equation}
\Phi_{\text{GP}}^{\textcolor{cyan}{\bullet}}(\rho_G) = \rho_G
\end{equation}

As we saw above, Gibbs-preserving maps can be time/phase-covariant. In this case, the fixed state $a_{\text{PC}}^* = r_G = \frac{\tau_z}{1 - \lambda_z}$. Operationally, the map parameters can be chosen such that the combination gives $r_G$. However, since the only constraint on Gibbs-preserving maps is that they preserve the Gibbs state, we can have more general dynamical maps as well. The most general Gibbs-preserving map can be derived by assuming a fully parameterized map $\Phi$ that obeys the above constraint.

Solving this constraint for the affine form of the dynamical map gives
\begin{equation}
\Phi_{\text{GP}}^{\textcolor{cyan}{\bullet}} = \begin{pmatrix}
1 & 0 & 0 & 0 \\
-T_{13}r_G & T_{11} & T_{12} & T_{13} \\
-T_{23}r_G & T_{21} & T_{22} & T_{23} \\
(1-T_{33})r_G & T_{31} & T_{32} & T_{33}
\end{pmatrix}
\end{equation}
Under this map, the vector $[1, 0, 0, r_G]$ gets mapped to itself. Therefore, a general Gibbs-preserving map need not be phase-covariant and can generate coherences and mix populations and coherences under evolution.
\section{Minimal qubit system}
\label{Sec:min_model}
To investigate the role of quantum resources in thermodynamic operations, we present simple models of two- and three-qubits for minimal Gibbs-preserving constructions \cite{Fauseweh2024,roadmap}.
\subsection{Two-Qubit System-Environment Model}
We first consider a system qubit S interacting with an environment qubit E. A general bipartite state can be decomposed into separable and non-separable parts:
\begin{equation}
\rho_{SE} = \rho_S \otimes \rho_E + \chi_{SE}
\end{equation}
For a two-qubit system, this can be written in the Pauli basis as:
\begin{equation}
\rho_{SE} = \frac{1}{4}\left(\sum_{i=0}^3(a_i \sigma_i\otimes\mathbf{1} + b_i \mathbf{1}\otimes\sigma_i)+\sum_{i,j=1}^3 c_{ij}\sigma_i\otimes\sigma_j\right)
\end{equation}
where $\sigma_0 = \mathbf{1}$, $\sigma_1 = \sigma_X$, $\sigma_2 = \sigma_Y$, $\sigma_3 = \sigma_Z$, and $\vec{a} = (a_1, a_2, a_3)$ and $\vec{b} = (b_1, b_2, b_3)$ are the Bloch vectors for the system and environment qubits respectively. The correlation matrix elements $c_{ij}$ characterize the quantum correlations between system and environment.

The two-qubit Hamiltonian is:
\begin{align*}
H &= H_S + H_E + H_{SE} \\
&= \frac{h_1}{2}\sigma_Z \otimes \mathbf{1} + \frac{h_2}{2}\mathbf{1}\otimes\sigma_Z + \frac{J}{2}(\sigma_X\otimes\sigma_X + \sigma_Y\otimes\sigma_Y)
\end{align*}
The condition for energy conservation $[H_{SE}, H_S + H_E] = 0$ is achieved when $h_1 = h_2$, ensuring that the interaction Hamiltonian commutes with the free Hamiltonian and preserves the U(1) symmetry generated by $G = \sigma_Z \otimes \mathbf{1} + \mathbf{1} \otimes \sigma_Z$.
\subsection{Three-Qubit Minimal Construction}
We also consider a three-qubit system with system (S), environment (E), and resource (R) qubits. The total Hamiltonian is:
\begin{equation}
H_{\text{tot}} = J\sum_{ij=SE,SR} (\sigma_{X,i}\sigma_{X,j} + \sigma_{Y,i}\sigma_{Y,j}) + h\sum_{i=S,E,R}\sigma_{Z,i}
\end{equation}
This Hamiltonian maintains energy conservation while allowing the system to interact with both environment and resource qubits through XX interactions. The Bloch vector for the system is $\vec{a} = (a_1, a_2,a_3)$, for the environment is $\vec{b} = (b_1, b_2,b_3)$, and for the resource qubit is $\vec{f} = (f_1, f_2,f_3)$

\begin{figure}[htbp]
\centering
\begin{tikzpicture}[scale=0.8]
    \definecolor{pastelblue}{RGB}{174, 198, 207}
    \definecolor{pastelgreen}{RGB}{209, 233, 192}
    \definecolor{pastelpink}{RGB}{255, 185, 170}
    
    \draw[fill=pastelgreen, draw=black, thick] (-3,0) circle (1) node {$\rho_R$};
    \node[below] at (-3,-1.3) {Resource};
    
    \draw[fill=pastelblue, draw=black, thick] (0,0) circle (1) node {$\rho_S$};
    \node[below] at (0,-1.3) {System};
    
    \draw[fill=pastelpink, draw=black, thick] (3,0) circle (1) node {$\rho_E$};
    \node[below] at (3,-1.3) {Environment};
    
    \draw[thick, <->] (-2,0) -- (-1,0);
    \node[above] at (-1.5,0.2) {$H_{SR}$};
    
    \draw[thick, <->] (1,0) -- (2,0);
    \node[above] at (1.5,0.2) {$H_{SE}$};
\end{tikzpicture}
\caption{Schematic of the three-qubit system for minimal Gibbs-preserving construction. The system qubit S interacts with both environment E and resource R through XY-type interactions. The S is the target qubit whose dynamics we will study, E is the environment qubit at some reference thermal temperature given by the z-Bloch vector as $b_3$ and R is the resource qubit that supplies coherence ($f_1, f_2 \neq 0$) needed for symmetry breaking}
\label{fig:three_qubit_schematic}
\end{figure}
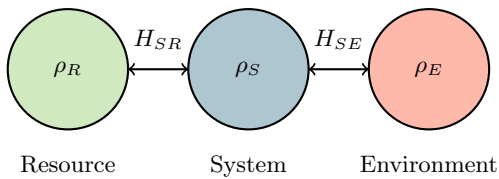
We show that a two-qubit system with local coherence in one qubit is sufficient to break phase covariance but cannot generate Gibbs-preserving maps unless global coherences are introduced. The three-qubit system is a minimal system where the addition of the resource qubit with local coherence provides the extra degree of freedom needed to satisfy the precise constraint relationships to get Gibbs-preserving dynamics while maintaining energy conservation.

\section{Results}
\label{Sec:Results}
\subsection{Fundamental Constraints: Coherence Requirements for Non-Unital Maps}
The U(1)-symmetry has direct implications for the Bloch vector map, particularly for the shift vector $\vec{\tau}$. Our first main result establishes the fundamental constraints imposed by this symmetry:
\begin{theorem}
\label{thm:no_coherence}
U(1)-symmetric unitary dynamics cannot generate local quantum coherence from initial states that are diagonal in the computational basis.
\end{theorem}
See Appendices \ref{App1:Unitarystrings} and \ref{App2:Nogo} for the complete proof based on charge conservation of Pauli strings under U(1) evolution. Using the Pauli operator decomposition, we can show that U(1) operators for N-qubits can be decomposed into Pauli strings such that each string has an equal number of $\sigma^+$ and $\sigma^-$ operators. Equivalently, we can say that out of the $4^N$ terms, a general U(1) operator only requires terms where the number of $\sigma_x$ plus number of $\sigma_y$ is even (See Appendix \ref{App1:Unitarystrings}), and these U(1) operations do not change the charge of a Pauli string. For a system that only has diagonal states, the charge of each Pauli string in the decomposition of the total density matrix is +1. To generate local coherence, the evolved density matrix must have terms proportional to $O_x = \sigma_x\otimes I\otimes I\ldots \otimes I$ with charge $C(O_x)=-1$. But charge is conserved under U(1) and therefore such a transformation forbidden starting from Pauli strings with C=+1 (See Appendix \ref{App3:Paulicharge}). 
\begin{corollary}
\label{cor:shift_constraint}
For qubit systems evolving under U(1)-symmetric dynamics with diagonal environment states $\rho_E$, the resulting dynamical map must satisfy $\tau_x = \tau_y = 0$ in the affine representation $\vec{a}' = \vec{\tau} + T\vec{a}$.
\end{corollary}
This result implies, furthermore, that just the presence of coherence alone is not enough. The coherence terms in the Pauli decomposition of the density matrix must have $C(\chi_{SE} + I\otimes \rho_E) = -1$. Only then can the unitary operations generate dynamics for the subsystem such that it has local coherences. 
For example, a system environment coherence of the form $\sigma_x\otimes\sigma_y\otimes I \ldots \otimes I$ cannot evolve to a term like $\sigma_x\otimes I \otimes I \ldots \otimes I$ under U(1) symmetric operations.\\

 However, if say at least one of the qubits in the environment $\rho_E$ possesses pre-existing coherences such that its Bloch vector is $\vec{b} = (b_1, b_2, b_3)$ corresponding to local coherence in the qubit, then a general U(1) operation can leverage the resource and generate coherences in the system qubit state such that $\tau_{x,y} \neq 0$. We show explicit calculation of this for a two qubit system evolving with a general U(1) unitary in Appendix \ref{App4:dynmapunderU}. 
This brings us to the second result:
\begin{theorem}
\label{thm:coherence_seed}
Non-zero shift components $\tau_x, \tau_y \neq 0$ in qubit dynamical maps under U(1)-symmetric evolution require pre-existing coherences in the environment or coherences between the system and the environment that couple states with unequal energy states.
\end{theorem}
See Appendices \ref{App2:Nogo} and \ref{App3:Paulicharge} for the proof showing how environmental coherence provides the odd-parity Pauli string components necessary to break the charge-conservation constraint.
Beyond the environment's initial state, any initial system-environment correlations, represented by $\chi_{SE}$, also contribute to the elements of the dynamical map. The term $\text{Tr}_E(U\chi_{SE}U^\dagger)$ can influence both the shift vector $\vec{\tau}$ and the transformation matrix $T$, potentially leading to more complex map structures even if $\rho_E$ is diagonal. We will now adopt these to describe the resources needed to go beyond phase-covariance.

\subsection{Two-Qubit Analysis: Breaking Phase-Covariance with Environmental Coherence}

For the two-qubit system with Hamiltonian $H = \frac{h_1}{2}\sigma_Z \otimes \mathbf{1} + \frac{h_2}{2}\mathbf{1}\otimes\sigma_Z + \frac{J}{2}(\sigma_X\otimes\sigma_X + \sigma_Y\otimes\sigma_Y)$ and energy conservation corresponding to the U(1) symmetry condition $h_1 = h_2$, we analyze how environmental coherence affects the resulting dynamical maps in Appendix \ref{App5:fulldynmap}. 

When the environment is diagonal ($b_1 = b_2 = 0$) and uncorrelated ($c_{ij} = 0$), Theorem~\ref{thm:no_coherence} constrains the dynamics to phase-covariant maps. The resulting map has the form:
\begin{equation}
\Phi_{PC} ^{\textcolor{magenta}{\bullet}}= \begin{pmatrix}
1 & 0 & 0 & 0 \\
0 & \frac{1}{2}c_\theta c_{h_+} & -\frac{1}{2}c_\theta s_{h_+} & 0 \\
0 & \frac{1}{2}c_\theta s_{h_+} & \frac{1}{2}c_\theta c_{h_+} & 0 \\
\frac{1}{2}b_3 s_\theta^2 & 0 & 0 & \frac{c_\theta^2}{2}
\end{pmatrix}
\end{equation}
where $\theta = J$, $h_+ = 2h$, $s_\theta = \sin(\theta),c_\theta = \cos(\theta),s_{h_+} = \sin(h_+),c_{h_+} = \cos(h_+),$and $s_{2\theta} = \sin(2\theta)$. The form of the map confirms $\tau_x = \tau_y = 0$ and restricting dynamics to isotropic scaling of coherence components. Even if the initial state had correlations of the form $\sigma_x\otimes\sigma_x+\sigma_y\otimes\sigma_y$ or $\sigma_x\otimes\sigma_y-\sigma_y\otimes\sigma_x$, they do not break to phase-covariance in the dynamical map since these strings, despite being coherences, have C = +1\ref{App5:fulldynmap}.

When the environment possesses coherence ($b_1, b_2 \neq 0$ in $\rho_E = \frac{1}{2}(I + b_1\sigma_X + b_2\sigma_Y + b_3\sigma_Z)$), the dynamics can lose the phase-covariant structure. The resulting map becomes:
\begin{equation}
\bar{\bar{\Phi}}^{\textcolor{black}{ \bullet}} = \begin{pmatrix}
1 & 0 & 0 & 0 \\
0 & \frac{1}{2}c_\theta c_{h_+} & -\frac{1}{2}c_\theta s_{h_+} & \frac{1}{2}s_\theta(b_1 s_{h_+} + b_2 c_{h_+}) \\
0 & \frac{1}{2}c_\theta s_{h_+} & \frac{1}{2}c_\theta c_{h_+} & \frac{1}{2}s_\theta(b_2 s_{h_+} - b_1 c_{h_+}) \\
\frac{1}{2}s_\theta^2 b_3 & -\frac{1}{4}b_2 s_{2\theta} & \frac{1}{4}b_1 s_{2\theta} & \frac{c_\theta^2}{2}
\end{pmatrix}
\end{equation}
However, while environmental coherence breaks phase-covariance, the presence of local-coherences in a system evolving with XX Hamiltonian, alone cannot achieve classes of maps like Gibbs-preserving dynamics with $h_1 = h_2$. This limitation is specific to XX interaction Hamiltonian and the resulting unitary structure. For this class of interactions, local coherence in the environment qubits is insufficient; global coherence between two qubits is needed for Gibbs-preservation. However, in Appendix \ref{App4:dynmapunderU}, we show that a general U(1)-symmetric unitary with local coherence in one qubit present can achieve dynamics beyond phase-covariance like Gibbs-preservation, demonstrating that the principle holds while highlighting the need for specific parameter tuning or additional resources for full Gibbs-preservation.

\subsection{Three-Qubit Analysis: Achieving Gibbs-preservation}
The three qubit system with the XX Hamiltonian provide the minimal system for Gibbs-preserving dynamics with only local coherence in the resource qubit. Under no constraints on the Bloch vector parameters of the resource qubit we get the following dynamical map. 
\begin{equation}
\Phi_{GP}^{\textcolor{cyan}{ \bullet}} =\begin{pmatrix}
1 & 0 & 0 & 0 \\
B_1\Omega_1 & A\phi_+ & A\phi_- & B_2\Omega_1 \\
B_1\Omega_2 & A\phi_- & A\phi_+ & B_2\Omega_2 \\
\frac{(b_3+f_3)s_{2J}^2}{2} & -\frac{f_2 s_{4J}}{2\sqrt{2}} & \frac{f_1 s_{4J}}{2\sqrt{2}} & c_{2J}^2
\end{pmatrix}
\end{equation}
Here, we have redefined some variables for cleaner representation. The compact notations represent $J' = \sqrt{2} J$, $s_J = \sin(\sqrt{2}J)$, $c_J = \cos(\sqrt{2}J)$, $s_{2J} = \sin(2\sqrt{2}J)$, $c_{2J} = \cos(2\sqrt{2}J)$, $s_{4J} = \sin(4\sqrt{2}J)$, $c_{4J} = \cos(4\sqrt{2}J)$, $s_{2h} = \sin(2h)$, and $c_{2h} = \cos(2h)$. We further define $A = 4(b_3 f_3 + 1)c_{2J} - (b_3 f_3 - 1)(c_{4J} + 3)$, $\Omega_1 = f_1 s_{2h} + f_2 c_{2h}$, $\Omega_2 = f_2 s_{2h} - f_1 c_{2h}$, $B_1 = \sqrt{2}b_3 s_J^3 c_J$, $B_2 = \sqrt{2}s_J c_J^3$, $\phi_+ = \frac{1}{8}\cos{2 h}$, and $\phi_- =\frac{1}{8}\sin{2 h}$.
\begin{theorem}[Minimal Gibbs-Preserving Construction]
\label{thm:minimal_construction}
Gibbs-preserving dynamics $\Phi_{GP}(\rho_G) = \rho_G$ can be achieved using exactly three qubits with the constraint relationships:
\begin{align}
b_3 s_J^2 &= -r_G c_J^2 \label{eq:constraint1}\\
f_3 &= 2r_G - b_3 \label{eq:constraint2}\\
\implies J &= \frac{1}{\sqrt{2}}\arctan\left[\sqrt{\frac{-r_G}{b_3}}\right] \label{eq:constraint3}
\end{align}
where $b_3$ characterizes the environment state, $f_3$ the resource state, and $J$ the interaction strength and $r_G$ is the z-component of the Bloch vector characterizing the Gibbs state $\rho_G$.
\end{theorem}
See Appendix \ref{App5:fulldynmap} for the complete derivation showing how these constraints ensure $\Phi_{GP}(\rho_G) = \rho_G$.
 Unlike just environmental coherence, which only breaks phase-covariance, Gibbs-preservation demands: (1) Environmental coherence to enable non-trivial rotations, (2) Resource qubit coherence ($f_1, f_2 \neq 0$) for additional symmetry breaking, and (3) Fine-tuned correlations satisfying the constraint relationships in Eqs.~(\ref{eq:constraint1})-(\ref{eq:constraint3}) to achieve the exact cancellation needed for Gibbs-preservation.

 The thermal reference with parameter $b_3$, the resource qubit supplies the coherence needed to break phase covariance, the interaction strength $J$ must be precisely tuned to achieve the right balance between thermal and coherent effects, and the system of constraints ensures that the global dynamics remain energy-conserving while allowing local Gibbs-preservation. The constraint system admits solutions in two distinct regions, characterized by the signs of $b_3$ and $r_G$, with the resource parameter $f_3$ determined by the linear relationship $f_3 = 2r_G - b_3$. 

\begin{figure}[htbp]
\centering
\includegraphics[width=\linewidth]{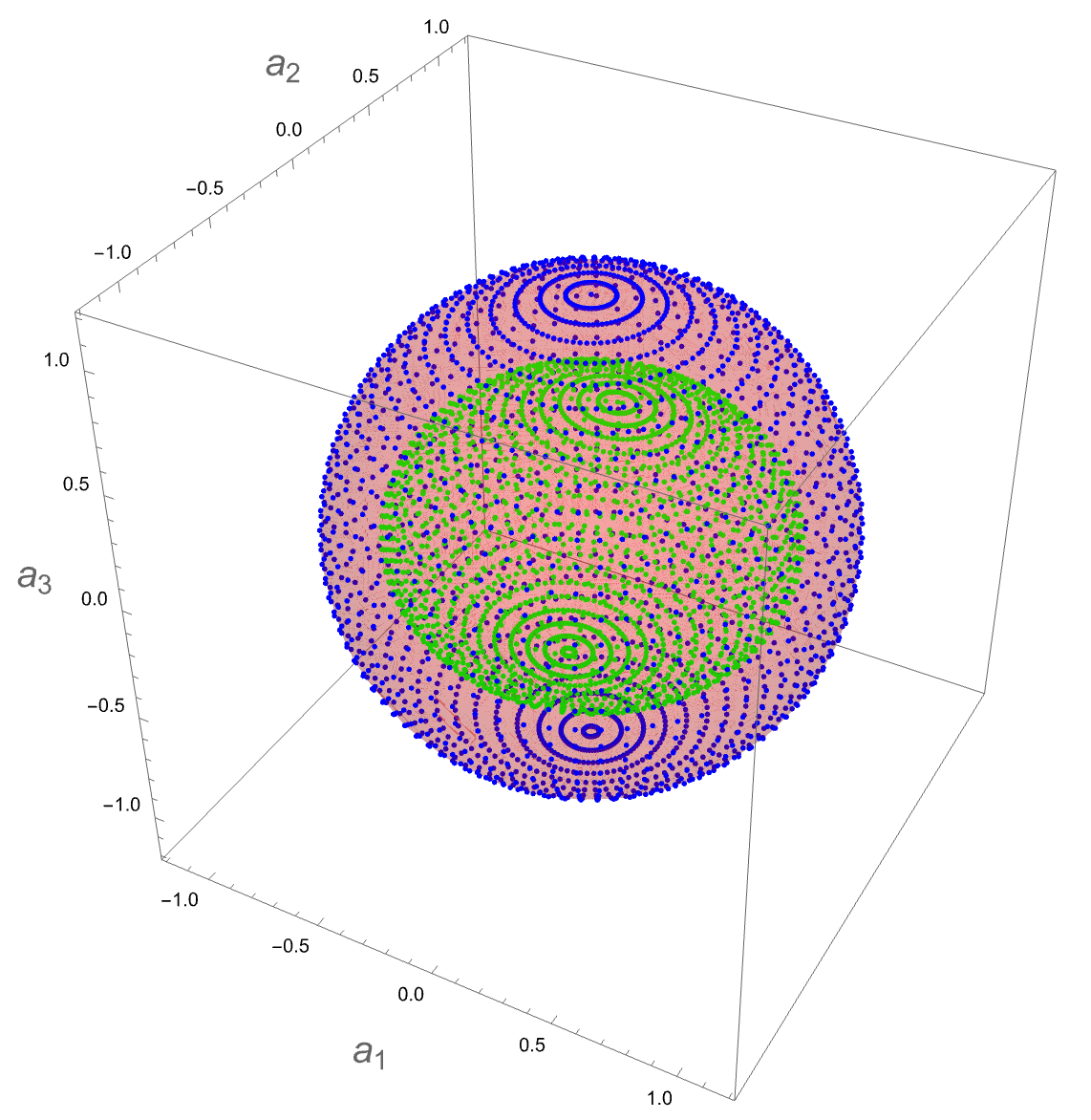}
\caption{Bloch sphere representation of the Gibbs-preserving map's action. Initial states (blue) are mapped to final states (green) under the map, demonstrating its contractive nature. Parameters: $b_3 = 0.3$, $h=\frac{\pi}{4}$, $J=0.25$, $f_1 =0.2$, $f_2= 0.1$.}
\label{fig:CPmapbloch}
\end{figure}

\subsection{Thermodynamic Advantages}

It has been shown that Gibbs-preserving maps offer thermodynamic advantages over phase-covariant thermal operations \cite{Faist_2015}. Here, we briefly quantify some advantages of the maps presented in the paper. First, we show the relative entropy difference between coherent states acted upon by the maps compared to the thermal state characterized by the Bloch vector $r_{th} =(0,0,r_G)$:
\begin{equation}
\Delta D = D(\Phi[\rho]||\rho_{th})
\end{equation}
This quantity is related to the extractable work from a system. Therefore, if $D(\Phi_{GP}[\rho]||\rho_{th})-D(\Phi_{PC}[\rho]||\rho_{th}) )= \delta >0$ then we conclude that $\delta$ is proportional to the excess extractable work when a system evolves under Gibbs-preserving dynamics vs thermal operations \cite{Kolchinsky_2025}. For both $\Phi_{PC}$ and $\Phi_{GP}$, the Bloch vector corresponding to the thermal state (shown as dotted line in Figure \ref{fig:DeltaD}) is an inert state $\Phi[\rho_{th}] = \rho_{th}$ and correspondingly the extractable work is null. For the map without the constraint for Gibbs-preservation or Phase-covariance, the inert state is not $\rho_{th}$.
\begin{figure}[htbp]
\centering
\includegraphics[width=\columnwidth]{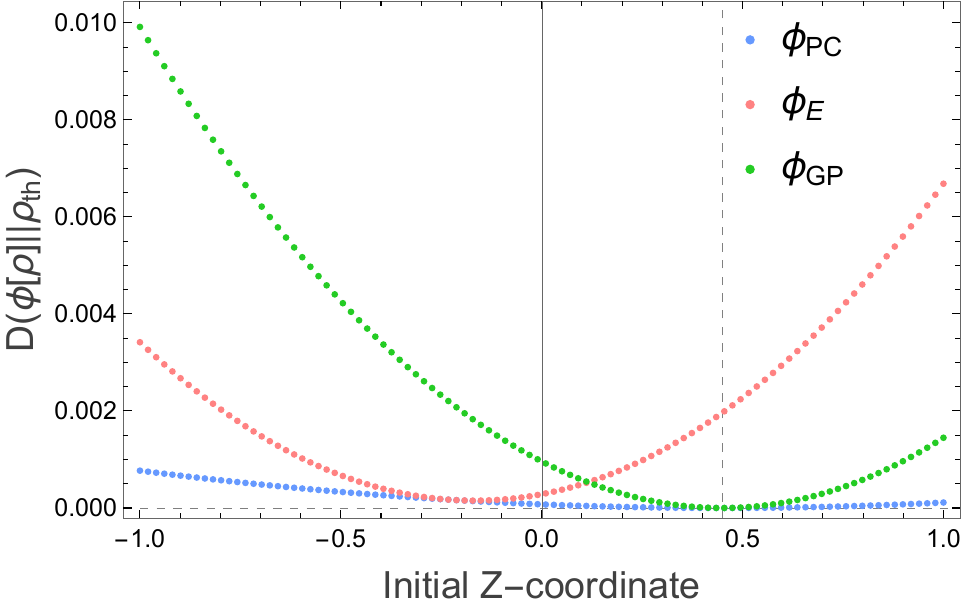}
\caption{Relative entropy difference ($\Delta D$) between incoherent states evolving under maps and the reference thermal state. The maps considered are the Gibbs-preserving map (green) ($f_1,f_2\neq 0$),  phase-covariant map (blue, $f_1=f_2= 0$) and a map with global energy conservation (pink) as a function of the initial incoherent state z-Bloch vector component. Higher $D$ values indicate a greater distance from thermal state and more extractable. Parameters: $J=0.5$ $b_3 = 0.3$, $r_G = 0.45$ $h=\frac{\pi}{4}$, $f_1 =0.2$, $f_2= 0.1$.}
\label{fig:DeltaD}
\end{figure}
The excess extractable work achievable through Gibbs-preserving dynamics demonstrates the concrete thermodynamic benefit of accessing the broader class of transformations. Since coherence generation is the distinguishing factor between the maps, we also show the coherence generation capacity of the maps starting from pure, incoherent state $\rho_0 = |0\rangle \langle 0|$. We consider the repeated action of the map on the initial state such that $\rho(n) = \Phi^n [\rho_0]$ for $n=2$ and show its effect on the coherence. We also present the rate at which correlations can decay under the action of the map starting from the incoherent state $\rho_+ = |+\rangle\langle +|$. 
\begin{figure}[htbp]
    \centering
    \begin{subfigure}{\columnwidth}
        \centering
        \includegraphics[width=\columnwidth]{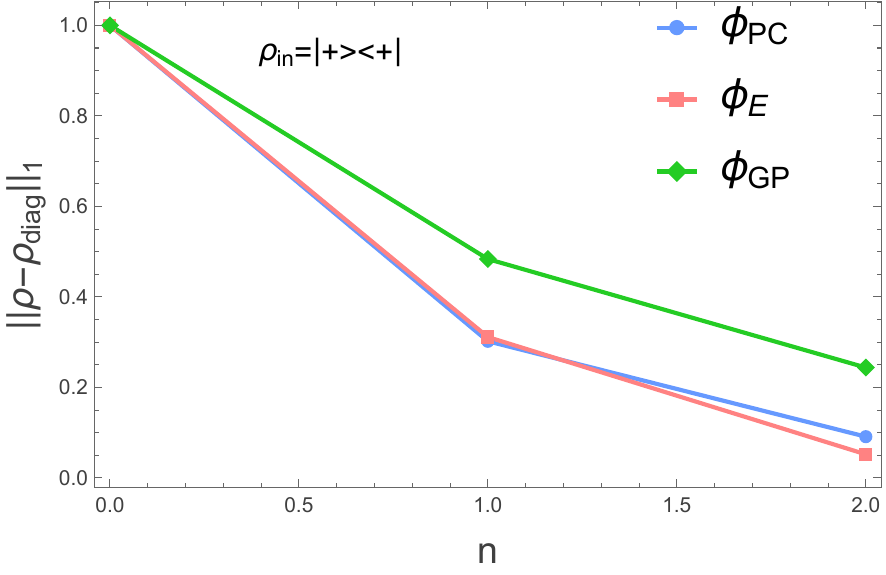}
        \label{fig:cpohdec}
    \end{subfigure}
    \vspace{1em} 
    \begin{subfigure}{\columnwidth}
        \centering
        \includegraphics[width=\columnwidth]{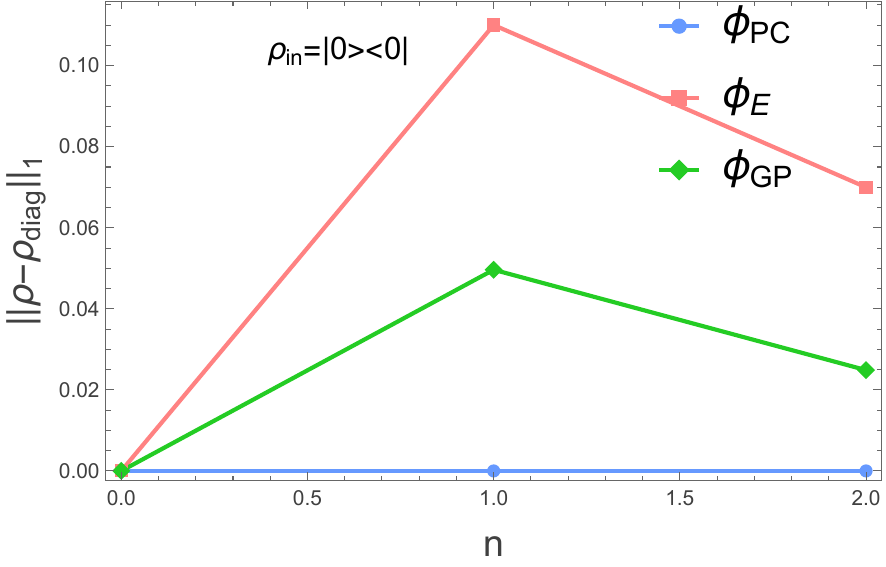}
        \label{fig:cohgen}
    \end{subfigure}
    \caption{The L1 norm between the of the off-diagonal part of the qubit density matrix under repeated action of the map for n =2 steps. The top figure shows the decay of correlations starting from the initial state with maximum coherence $\rho_{in}= |+\rangle\langle+|$ and the bottom figure shows the growth of correlations starting from the initial state with no coherence $\rho_{in}= |0\rangle\langle0|$ and the The maps considered are the Gibbs-preserving map (green) ($f_1,f_2\neq 0$),  phase-covariant map (blue, $f_1=f_2= 0$) and a map with global energy conservation (pink) as a function of the initial incoherent state z-Bloch vector component. Parameters: $J=0.5$ $b_3 = 0.3$, $r_G = 0.45$ $h=\frac{\pi}{4}$, $f_1 =0.2$, $f_2= 0.1$.}
    \label{Cohgenanddec}
\end{figure}
These plots show that the coherence decay is slower under $\Phi_{GP}$ compared to $\Phi_{PC}$. Starting from a coherent state $\Phi_{GP}$ can develop coherences whereas $\Phi_{PC}$ corresponding to thermal operations cannot. Lastly, we show the rate of convergence of the map. Starting from the initial state $\rho_0 = |0\rangle \langle 0|$, we apply the map repeatedly and compute the L1-norm between the density matrices at subsequent steps until the norm return 0.  We can see that the phase-covariant map converges to the thermal state quickest while the Gibbs-preserving map is the slowest. The general map is intermediate in terms of rate of convergence.  
These advantages are non-trivial and depend on the specific parameter choices, showing that there exist regimes where the investment in quantum resources yields measurable thermodynamic returns. The numerical analysis confirms that both information-theoretic and work extraction benefits are achievable across the feasible parameter space, validating the practical significance of our theoretical framework. These results establish that the systematic approach to overcoming symmetry constraints through quantum resource engineering provides genuine advantages over classical thermal approaches, with implications for designing efficient quantum thermal devices.
\begin{figure}[htbp]
    \centering
\includegraphics[width=0.8\columnwidth]{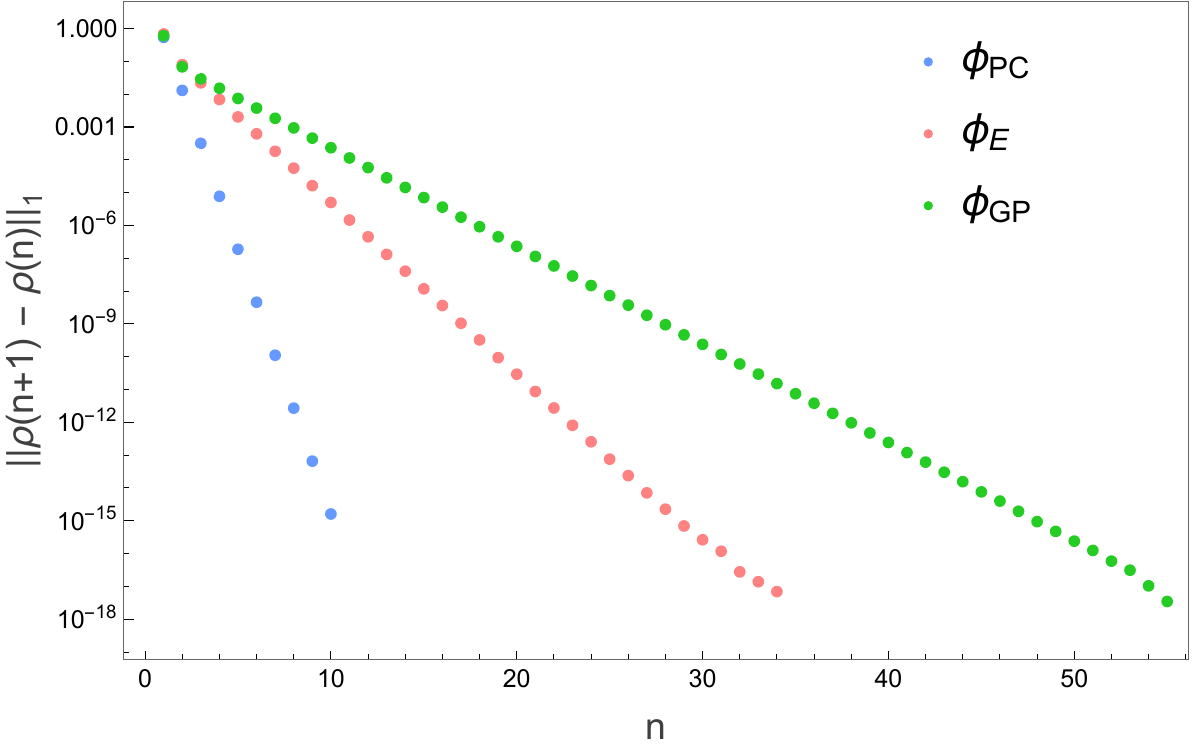}
    \caption{The rate of convergence of the maps when acting on $\rho_{in} = |0\rangle\langle0|$. We plot the L1-norm of difference between the qubit density matrix upon subsequent action of the map for n steps until the norm is 0. The convergence is fasted for $\Phi_{PC}$ (blue) followed by $\Phi_E$ (pink) and is slowest for $\Phi_{GP}$ (green). Parameters: $J=0.5$ $b_3 = 0.3$, $r_G = 0.45$ $h=\frac{\pi}{4}$, $f_1 =0.2$, $f_2= 0.1$.}
    \label{fig:Convergence}
\end{figure}
\section{Conclusions}
We have established a comprehensive framework for understanding how quantum symmetries constrain thermodynamic transformations and how these constraints can be systematically overcome through quantum resources. Our key findings include:
\begin{enumerate}
\item A No-go theorem showing that U(1) symmetric operations cannot generate local coherence from thermal states
\item A complete characterization of the resource hierarchy from thermal operations to Gibbs-preserving maps to general energy-preserving maps
\item An explicit construction demonstrating how minimal quantum resources can achieve Gibbs-preservation
\item Quantitative measures showing the thermodynamic advantages of accessing broader transformation classes
\end{enumerate}
From a resource theory perspective, our work identifies the precise resources needed to move beyond thermal operations: environmental coherence and fine-tuned correlations. This provides a roadmap for understanding when and how quantum resources can enhance thermodynamic performance \cite{roadmap}.\\
Our hierarchy of thermodynamic operations (Fig. 1) complements characterization of covariant Gibbs-preserving maps by explicitly demonstrating the symmetry-breaking mechanisms required to transition between operation classes. While previous work established that free energy governs state convertibility with correlated catalysts, we show precisely how environmental coherence and correlation resources enable this symmetry breaking \cite{shiraishi2024quantumthermodynamicscoherencecovariant}.
Looking forward, this work opens several avenues for exploration, including the extension to larger systems, the investigation of other symmetry classes, and the development of optimal protocols for resource-efficient thermodynamic transformations. The interplay between symmetry, coherence, and thermodynamics promises to remain a rich source of fundamental insights in quantum physics.
\section*{Acknowledgments}
The author thanks Marcus Huber, Sarah Shandera, and Tommy Chin for valuable comments and discussions throughout the development of this work.
\bibliography{main}
\begin{widetext}
\appendix

\section{U(1) symmetric qubit unitaries }
\label{App1:Unitarystrings}
This appendix provides a mathematical characterization of unitary operators that commute with the U symmetry generator $G = \sum_i \sigma_Z^{(i)}$. As established in the main text, such unitaries exhibit a block-diagonal structure: $U=\mathcal{U}_0 \oplus \mathcal{U}_1 \oplus \ldots \oplus \mathcal{U}_N$, where each $\mathcal{U}_m$  corresponds to an independent unitary matrix within the subspace with $m$ qubits is $|1\rangle$ state, and is of dimensions $^NC_m\times ^NC_m$. Any element within this subspace can be written as
\begin{equation}
\mathcal{U}_m = \sum_{\alpha,\beta \in \mathcal{S}_m} u_{\alpha\beta} |\alpha\rangle\langle \beta|
\end{equation}
where $\mathcal{S}_m$ denotes all computational basis states with magnetization $m$. Each matrix element $u_{\alpha\beta}$ contributes to the Pauli expansion of $U$. States with $m$ excitations can be written as $|\alpha\rangle = |s_1, s_2, \ldots, s_N\rangle$ where $s_i \in {0, 1}$ and $\sum s_i = m$. In the ladder operator expansion, this state is $|\alpha \rangle = \left(\bigotimes_{i=1}^N O_i^{(\alpha)} \right) |00\ldots 0\rangle$ each $O_i ^{(\alpha)}\in\{\sigma^+,I,\sigma_Z\}$ with m of them being $\sigma^+$. Therefore, any matrix element inside the subspace with $m$ up states is given by 
\begin{align}
    \mathcal{U}_m &= \sum_{\alpha,\beta \in \mathcal{S}_m} u_{\alpha\beta}|\alpha\rangle\langle \beta|\\&=\sum_{\alpha,\beta \in \mathcal{S}_m} u_{\alpha\beta}\left(\bigotimes_{i=1}^N O_i^{(\alpha)}\right) |00\ldots0\rangle \langle 00\ldots 0|\left(\bigotimes_{j=1}^N O_j^{\dagger(\beta)}\right)
    \\
    &= \sum_{\alpha,\beta \in \mathcal{S}_m} u_{\alpha\beta}\left(\bigotimes_{i=1}^N O_i^{(\alpha)}\right) . Q. \left(\bigotimes_{j=1}^N O_j^{\dagger(\beta)}\right)
\end{align}
where $Q = |00\ldots0\rangle\langle00\ldots0|$ is the reference state. This projector onto the all-zero state can be expanded in terms of Pauli $\sigma_Z$ and identity operators for $N$ qubits:
$$ Q = \frac{1}{2^N} \bigotimes_{i=1}^N (I^{(i)} - \sigma_Z^{(i)}) $$
This expansion shows that $Q$ is a sum of $2^N$ Pauli strings, where each string $\tilde{O}_j$ is a tensor product of $I$ and $\sigma_Z$ operators.

 By construction the number of $ \sigma^+$ and $\sigma^-$ operators for any $|\alpha\rangle \langle \beta|$ for the unitary within a subspace is equal; however, the index at which they occur may be different. When expanding the product $\left( \otimes O_i ^{(\alpha)} \right) Q \left( \otimes O_i ^{(\beta)} \right)$ into Pauli strings, several cases arise for each qubit index. If $O_j ^\alpha =\sigma^+$ and $O_j ^{\dagger\beta }=\sigma^-$ occur on the same index j, then, depending on whether $\tilde{O}_j$ is identity or $\sigma_z$, the resultant operator on index j is either $\pm I$, but crucially cannot be $\sigma^\pm$. If the operator on the left has $\sigma^+$ on the j index and the operator on the right has $I$ or $\sigma_z$ on the j index then depending on the operator on $\tilde{O}_j$ the resultant operator on index j is $\pm \sigma^+$. If all the three operators have I or $\sigma^z$ on the j index then the resultant will also be either I or $\sigma^z$. Furthermore, terms of the form $\sigma^+_j \sigma^-_k$ where $j \neq k$ contribute to $\sigma_X$ or $\sigma_Y$ components across different sites. \\ 
 The total number of $\sigma_X/\sigma_Y$ operators arising from such products within a charge-conserving term is always even. This can be seen as follows. Operators of the form $O^\alpha = \otimes O _i ^\alpha$ (right, $O^\beta = \otimes O_j ^\beta$) contain $m$ $\sigma^{+(-)}$ at $m$ out of $N$ indices. Let $k$ out of $m$ be indices where the operators $O^\alpha$ and $O^\beta$ concurrently have $\sigma^+$ and $\sigma^-$ respectively. For $m-k$ sites, $\sigma^+$ in $O^\alpha$ $\sigma^-$ on $O^\beta$ are on mismatched indices. This implies that of $N$: $(m-k)$ indices have operators $\pm\sigma^+$, $(m-k)$ indices have operators $\pm\sigma^-$ and the remaining indices have $I$ or $\sigma_Z$. This follows for all the terms in the expansion of $Q$ and the unitary within each block with fixed $m$. The full unitary operator is then the sum of all such possible operators.\\
The space of all operators on an $N$-qubit system has dimension $4^N$. The number of linearly independent Pauli operators that commute with $G=\sum_i \sigma_Z^{(i)}$ are precisely those composed only of $\sigma_Z$ and $I$ operators, or products of an even number of $\sigma_X$ and $\sigma_Y$ operators. This means that for any such operator $O = \bigotimes_{i=1}^N P_i$, the number of $P_i \in {\sigma_X, \sigma_Y}$ must be even.
The total number of such charge-conserving Pauli strings, $N_G$, for $N$ qubits is given by:
$$ N_G = \sum_{j=0}^{\lfloor N/2 \rfloor} \binom{N}{2j} 2^{2j} 2^{N-2j} = 2^N \sum_{k=0}^{\lfloor N/2 \rfloor} \binom{N}{2j} $$
Here, $\binom{N}{2j}$ selects $2j$ qubits indices for $\sigma_X$ or $\sigma_Y$ operators. $2^{2j}$ accounts for the choice of $\sigma_X$ or $\sigma_Y$ for these $2k$ qubits. $2^{N-2j}$ accounts for the choice of $I$ or $\sigma_Z$ for the remaining $N-2j$ qubits.\\
The ratio of these allowed operators to all possible $4^N$ operators is $\frac{N_G}{4^N} = \frac{2^N \sum_{j=0}^{\lfloor N/2 \rfloor} \binom{N}{2j}}{4^N} = \frac{1}{2^N} \sum_{j=0}^{\lfloor N/2 \rfloor} \binom{N}{2j}$. The sum for $N>0$ is $2^{2N-1}$, implying the ratio is $1/2$. While half of all Pauli strings commute with $G$, the space of all possible unitaries is vast.\\
\section{Charge of Pauli strings under U(1) evolution}
\label{App3:Paulicharge}
\begin{proof}
We divide the Pauli operators into two groups of commuting with $\sigma_z$ operator $P_c: \{I, \sigma_z\}$ and non-commuting terms $P_{\bar c}: \{\sigma_x, \sigma_y\}$. These are not subgroups since multiplications can result in operators mixing with one another $P_{c} P_{\bar{c}} \in P_{\bar{c}}$ and $P_{\bar{c}} P_{\bar{c}} \in P_c$. Let us associate a charge, $C(\cdot)$, to operators $C(P_c) \to 1$ and $C(P_{\bar{c}}) \to -1$ to allow us to count the number of $P_{\bar{c}}$ terms in a Pauli string after action of unitary rotation. \\
The charge for a Pauli string is then given by the product of the charge from operator from each index. By definition, since U(1)
 symmetric rotations have even number of $\sigma_x$ and $\sigma_y$ operators, the charge for each Pauli string in the operator expansion of the unitary or its conjugate will be +1.\\
 Let us consider a Pauli string, S, of length N with 
 $$
S = P_1\otimes P_2\otimes \ldots\otimes P_N \quad \text{where each } C(P_j) = \pm 1
$$
The evolution of this operator is given by $S\to S' = U.S.U^\dagger$ and the charge of the operator is $C(S) = \prod_{j=1}^N C(S_j) = c \text{ where } c= \pm 1$
Let Pauli strings in the operator expansion of the U be of the form 
$$
L = L_1\otimes L_2\otimes \ldots\otimes L_N \quad \text{where each } C(L_j) = \pm 1, \text{ with an even number with charge } -1\text{'s}
$$
and the strings in the operator expansion of the $U^\dagger$ be of the form 
$$
R = R_1\otimes R_2\otimes \ldots\otimes R_N \quad \text{where each } C(R_j) = \pm 1, \text{ with an even number with charge } -1\text{'s}
$$
The total charge of the string $L$ and $R$ is $C(L) = +1$. Then if,
$
C(L_j) = C(R_j), \text{ then } C(S_j) \text{ is unchanged}
$ and if
$
C(L_j) \ne C(R_j), \text{ then } C(S_j) \rightarrow -C(S_j)
$
Therefore, the resultant string \( S' \) has:
$$
C(S_j') =
\begin{cases}
C(S_j), & \text{if } C(L_j) = C(R_j) \\
- C(S_j), & \text{if } C(L_j) \ne C(R_j)
\end{cases}
$$
Let \( f \) be the number of indices where \( C(L_j) \ne C(R_j) \). Each of these indices contributes to a flipping of charge resulting in a total charge flip of  $C(S')=(-1)^f C(S)$. \\
Independently, let there be m indices where $C(L_i) = C(R_i) =-1 $. If m is even, then L and R are restricted to have only an even number of indices, $e_L$ and $e_R$, respectively, where the charge is -1. Since these are mismatched, the total number of mismatched indices will be $e_L + e_R$ which is necessarily even. If instead m is odd, then L and R are restricted to have only an odd number of indices, $o_L$ and $o_R$, respectively, where the charge is -1. Since these are mismatched, the total number of mismatched indices will be $o_L + o_R$ which is also necessarily even. \\
This shows that the number of mismatched indices has to be even, implying that $f$ is even. This implies that the change in charge under evolution is 
$$
C(S')= (-1)^f C(S) = (-1)^f c = c \hspace{0.3em}\text{ when $f$ is even}
$$
This proves that the charge of each Pauli string is independently conserved under the action of U(1) terms. \end{proof} 
\begin{corollary}
    N-body correlation terms between the system and environment that have even charge, or equivalently, an even number of $\sigma_x + \sigma_y$ operators, cannot contribute to the shift parameter of rotation matrix in the $x$ and $y$ direction in the dynamical map.   
\end{corollary}
\begin{proof}
    To contribute to the shift vector or the rotation matrix in the $x$ and $y$ direction in the dynamical map, the evolution of a Pauli string, S, has to be
    \begin{equation}
        S' = U.S.U \propto \sigma_x \otimes I\otimes I\ldots\otimes I + \text{other terms}
    \end{equation}
If the number of terms of $\sigma_x$ and $\sigma_y$ operators in the Pauli decomposition of S is even, then the  the charge of S is $C(S) = (-1)^{n_x+n_y} = +1$. But the charge of the desired term is $C(\sigma_x \otimes I\otimes I\ldots\otimes I) = \prod_i C(\sigma_x)C(I)^{N-1} = -1$ . For this to happen, $C(S)=-1 \to C(S') = +1 $ but this is forbidden by U(1) symmetric rotations and therefore cannot be generated. 
\end{proof}
\section{Proof for the necessity of coherence in environment to have non zero shift under U(1) dynamics }
\label{App2:Nogo}
This section rigorously examines the conditions under which quantum coherences can be generated in a system evolving under U(1)-symmetric dynamics. We consider the implications for a one-qubit system in contact with N-1 qubits in an incoherent state,  initially prepared in a product state where each local density matrix is diagonal in the computational basis. Such states arise when system is in contact with a larger system in thermal equilibrium, lacking coherences.\\
Let the initial state of the $N$-qubit system be $\rho = \bigotimes_{i=1}^N \rho_i$, where each the system qubit (index i = 1) is a general qubit state while the rest of the N-1 qubits are diagonal in the computational basis. This means $\rho_i = p_{i,0}|0\rangle\langle0| + p_{i,1}|1\rangle\langle1|$. Consequently, the total state $\rho$ can be expressed as a linear combination of product operators where the operator at index $i=1$ can be any Pauli operator but for all the other indices, the operators are either identity ($I_k$) or $\sigma_{Z,k}$ operators:
$$ \rho = \sum_j \alpha_j \tilde{O}_j $$
The total system evolves under a unitary $U$ that conserves the total excitation number, meaning it commutes with the generator of U symmetry, $G = \sum_i \sigma_{Z,i}$, i.e., $[U, G]=0$. As established in Appendix A, any Pauli string $P = P_1 \otimes \cdots \otimes P_N$ with a non-zero coefficient in the expansion of such a U(1)-symmetric unitary $U$ must contain an even number of $\sigma_X$ or $\sigma_Y$ operators across all $N$ qubits.\\
\begin{theorem}
Number-conserving dynamics, governed by a U(1)-symmetric unitary, cannot generate local quantum coherence from an initial product state that is diagonal in the computational basis.
\end{theorem}
\begin{proof}
First, let us consider that the density matrix of each of the qubits, including the system qubit, is initially diagonal in the computational basis. The evolution of the total density matrix is given by $U \rho U^\dagger = U \left( \sum_j \alpha_j \tilde{O}_j \right) U^\dagger = \sum_j \alpha_j U \tilde{O}j U^\dagger$.We need to examine the structure of $U \tilde{O} U^\dagger$ where each $\tilde{O}$ is a N-length  Pauli string with $I_k$ and $\sigma_{Z,k}$ operators.\\
Let $U = \sum_i c_{Li} L_i$ be the Pauli expansion of $U$, and $U^\dagger = \sum_{i} c_{Ri} R_i$ (where $R_i$ are also Pauli strings). A generic term in the expansion of $U \tilde{O} U^\dagger$ is of the form $c_{Li} c_{Rj} L_i \tilde{O} R_j$. For the evolved state to generate local coherences in any qubit density matrix, say the qubit at index $k$, the evolved density matrix must contain terms proportional to $\sigma_{X,k}$ or $\sigma_{Y,k}$ at that qubit index.
Consider the product $L_i \tilde{O} R_j$. The operator $\tilde{O}$ consists only of $I_k$ and $\sigma_{Z,k}$ for each qubit $k$. These operators do not change the parity of $\sigma_X$ or $\sigma_Y$ operators in a Pauli string. Therefore, the total number of $\sigma_X$ or $\sigma_Y$ operators in $L_i \tilde{O} R_j$ is determined solely by the combined number of $\sigma_X$ or $\sigma_Y$ operators in $L_i$ and $R_j$.\\
Since both $L_i$ and $R_j$ originate from the expansion of a U(1)-symmetric unitary, each must contain an even number of $\sigma_X$ or $\sigma_Y$ operators. Consequently, their product $L_i \tilde{O} R_j$ must also contain an even number of $\sigma_X$ or $\sigma_Y$ operators in total. However, generating a coherence term requires that the evolved operators have a term that has $\sigma_{X,j}$ at some index j and identity elsewhere. Since this would imply an odd number of $\sigma_X$ in the resulting Pauli-string, it directly contradicts the requirement that $L_i \tilde{O} R_j$ must have an even number of $\sigma_X$ or $\sigma_Y$ operators and hence cannot be generated.
\end{proof}
\begin{corollary}
     If each $\rho_i$ is diagonal, then the reduced state of any subsystem, $\rho_S' = \text{Tr}_E(U\rho U^\dagger)$, will also be diagonal in the computational basis for $S$.
\end{corollary}
The preceding no-go theorem on coherence generation relies on the absence of initial coherences in the system and environment. This constraint is lifted if the initial state contains pre-existing coherences.\\
Consider an initial state where one or more environment qubits possess coherences in the computational basis, for example, an environment qubit $j$ starting in the state $\rho_j = \frac{1}{2}(I_j + \beta \sigma_{X,j})$ (where $\beta$ is a real coefficient). In this scenario, the initial state $\rho$ will include terms like $\tilde{O}j$ that contain $\sigma_X$ or $\sigma_Y$ operators. For example, the term involving $\sigma_{X,j}$ in $\rho_j$ has an odd number of $\sigma_X/\sigma_Y$ operators (one, $\sigma_{X,j}$ itself).\\
When we analyze the evolution of such a term, $U (\ldots \otimes \sigma_{X,j} \otimes \ldots) U^\dagger$, we consider products $L_i (\ldots \otimes \sigma_{X,j} \otimes \ldots) R_j$. The total number of $\sigma_X/\sigma_Y$ operators in this product will now come from the product of (even from $L_i$) , (odd from $\sigma_{X,j}$) and (even from $R_j$), which results in an odd total number of $\sigma_X/\sigma_Y$ operators.\\
This odd parity allows for the creation of local coherences on other qubits (e.g., a system qubit) while maintaining overall consistency with the U-symmetry of $U$. Specifically, this odd number of $\sigma_X/\sigma_Y$ operators can manifest as a single $\sigma_X$ or $\sigma_Y$ term on a system qubit (e.g., $\sigma_{X,1} \otimes I_{\text{rest}}$), thus generating local coherence.

\section{Dynamical map under general U(1) unitary}
\label{App4:dynmapunderU}
Consider a general unitary operator $U$ acting on two qubits that commutes with the U(1) generator $G = \sigma_Z \otimes I + I \otimes \sigma_Z$. This constraint ensures energy conservation when the local Hamiltonians have matched energy scales. The most general form of such a unitary can be parametrized using five real parameters:
\begin{equation}
  \Phi=  \left(
\begin{array}{cccc}
 e^{i \text{$\varphi $0}} & 0 & 0 & 0 \\
 0 & e^{-i (\alpha +\text{$\varphi $1})} \cos \left(\frac{\theta }{2}\right) & -e^{i (\text{$\varphi $1}-\alpha )} \sin \left(\frac{\theta }{2}\right) & 0 \\
 0 & e^{-i (\text{$\varphi $1}-\alpha )} \sin \left(\frac{\theta }{2}\right) & e^{i (\alpha +\text{$\varphi $1})} \cos \left(\frac{\theta }{2}\right) & 0 \\
 0 & 0 & 0 & e^{i \text{$\varphi $2}} \\
\end{array}
\right)
\end{equation}
 This parametrization exhibits the characteristic block-diagonal structure of U(1) symmetric unitaries, with independent blocks acting on subspaces of fixed excitation number: the zero-excitation subspace $\{|00\rangle\}$, the one-excitation subspace $\{|01\rangle, |10\rangle\}$, and the two-excitation subspace $\{|11\rangle\}$.\\
 Let the initial states of the system and environment qubits be characterized by Bloch vectors $\vec{a} = (a_1, a_2, a_3)$ and $\vec{b} = (b_1, b_2, b_3)$, respectively. The evolution $\rho_{SE} \to U\rho_{SE}U^\dagger$ followed by partial trace over the environment yields the effective dynamical map on the system qubit. To present the resulting map in compact form, we introduce the following auxiliary variables: - $\xi = \frac{1}{2}(2\alpha - \varphi_0 - 2\varphi_1 + \varphi_2)$ - $\psi_+ = \alpha - \varphi_0 + \varphi_1$ - $\psi_- = \alpha + \varphi_1 + \varphi_2$ - $\chi = \frac{\varphi_0 + \varphi_2}{2}$ The dynamical map on the system qubit takes the affine form $\vec{a}' = \vec{\tau} + T\vec{a}$, where: 
 \begin{equation}
     \Phi = \begin{pmatrix} 1 & 0 & 0 & 0 \\ \tau_x & T_{11} & T_{12} & T_{13} \\ \tau_y & T_{21} & T_{22} & T_{23} \\ \tau_z & T_{31} & T_{32} & T_{33} \end{pmatrix}
      \end{equation}
      with matrix elements:  \begin{align*}
      \tau_x &= \sin\left(\frac{\theta}{2}\right) \sin(\chi) \left[b_1 \sin(\xi) + b_2 \cos(\xi)\right]\\
      \tau_y &= \sin\left(\frac{\theta}{2}\right) \sin(\chi) \left[b_2 \sin(\xi) - b_1 \cos(\xi)\right]\\
      \tau_z &= b_3 \sin^2\left(\frac{\theta}{2}\right)
      \end{align*}
      and the transformation matrix elements:
      \begin{align*}
      T_{11} = T_{22} &= \frac{1}{2}\cos\left(\frac{\theta}{2}\right)[(b_3 + 1)\cos(\psi_+) - (b_3 - 1)\cos(\psi_-)]\\
      T_{12} = -T_{21} &= \frac{1}{2}\cos\left(\frac{\theta}{2}\right)[(b_3 - 1)\sin(\psi_-) - (b_3 + 1)\sin(\psi_+)]\\
      T_{13} &= \sin\left(\frac{\theta}{2}\right) \cos(\chi) \left[b_1 \cos(\xi) - b_2 \sin(\xi)\right]\\
      T_{23} &= \sin\left(\frac{\theta}{2}\right) \cos(\chi) \left[b_1 \sin(\xi) + b_2 \cos(\xi)\right]\\
      T_{31} &= -\frac{1}{2}\sin(\theta)[b_1 \cos(2\varphi_1) + b_2 \sin(2\varphi_1)]\\
      T_{32} &= \frac{1}{2}\sin(\theta)[b_1 \sin(2\varphi_1) - b_2 \cos(2\varphi_1)]\\
      T_{33} &= \cos^2\left(\frac{\theta}{2}\right)\end{align*}.
      For the XX-type Hamiltonians studied in the main text, where $H_{SE} = \frac{\alpha}{2}(\sigma_X \otimes \sigma_X + \sigma_Y \otimes \sigma_Y)$ with matched local energy scales, the parameters simplify to $\varphi_0 = -\varphi_2$ and $\varphi_1 = \alpha = 0$. Under these constraints: - $\xi = -\frac{\varphi_0}{2}$ - $\psi_+ = -\varphi_0$, $\psi_- = \varphi_0$ - $\chi = 0$ This reduction yields the phase-covariant maps analyzed in Appendix \ref{App5:fulldynmap}, demonstrating how the general framework encompasses the specific cases relevant to quantum thermodynamics.\\
      This general parametrization reveals that breaking phase-covariance requires either: (i) Non-zero $b_1, b_2$ components enable non-trivial shift vectors $\tau_x, \tau_y$  or (ii) Unequal phase parameters $\varphi_0 \neq \varphi_2$ or $\alpha \neq 0$ that can generate coupling between populations and coherences. 

 \section{Full dynamical maps calculation
}
\label{App5:fulldynmap}
\begin{equation}
\Phi_E=\left(
\begin{array}{cccc}
 1 & 0 & 0 & 0 \\
 \frac{J s_\theta (c_{31} s_{h_+}+c_{32} c_{h_+})}{2 \theta } & \frac{1}{2} \left(c_\theta c_{h_+}-\frac{h_- s_\theta s_{h_+}}{\theta }\right) & -\frac{h_- s_\theta c_{h_+}+\theta c_\theta s_{h_+}}{2 \theta } & \frac{J s_\theta (b_{1} s_{h_+}+b_{2} c_{h_+})}{2 \theta } \\
 \frac{J s_\theta (c_{32} s_{h_+}-c_{31} c_{h_+})}{2 \theta } & \frac{h_- s_\theta c_{h_+}+\theta c_\theta s_{h_+}}{2 \theta } & \frac{1}{2} \left(c_\theta c_{h_+}-\frac{h_- s_\theta s_{h_+}}{\theta }\right) & \frac{J s_\theta (b_{2} s_{h_+}-b_{1} c_{h_+})}{2 \theta } \\
 \frac{s_\theta ^2 \left( 2J ^2 b_{3}+2 J h_- (c_{11}+c_{22})\right)+ J \theta (c_{21}-c_{12}) s_{2\theta}}{4 \theta ^2} & \frac{J s_\theta (b_{1} h_- s_\theta-b_{2} \theta c_\theta)}{2 \theta ^2} & \frac{J s_\theta (b_{1} \theta c_\theta+b_{2} h_- s_\theta)}{2 \theta ^2} & \frac{ \theta ^2+c_{2\theta}J ^2+h_-^2}{4 \theta ^2} \\
\end{array}
\right)
\label{eqn:GenMap}
\end{equation}
where $h_-=h_1=h_2, h_+=h_1+h_2,\theta=\sqrt{J^2 + h_- ^2},  s_\theta = \sin(\theta),c_\theta = \cos(\theta),s_{h_+} = \sin(h_+),c_{h_+} = \cos(h_+),$and $s_{2\theta} = \sin(2\theta)$.
\subsection{Thermal operations and Phase-covariant dynamics}
For phase-covariant dynamics, we impose the following: (i) $[H_{SE},H_S+H_E]\implies h_1=h_2=h\implies h_- = 0,  h_+ = 2h \text{ and } \theta=J$, (ii) $[H_E,\rho_E] = 0 \implies b_1 = b_2 = 0$ and (iii) The system and environment are initially uncorrelated $\implies c_{ij}=0 \hspace{0.2em} \forall i,j = 1,2,3 $. These simplifications result in the following map
\begin{equation}
    \Phi_{PC} ^{\textcolor{magenta}{\bullet}}=\left(
\begin{array}{cccc}
 1 & 0 & 0 & 0 \\
 0 & \frac{1}{2} c_\theta c_{h_+} & -\frac{1}{2} c_\theta s_{h_+} & 0 \\
 0 & \frac{1}{2} c_\theta s_{h_+} & \frac{1}{2} c_\theta c_{h_+} & 0 \\
 \frac{1}{2} b_{3} s_{\theta} ^2 & 0 & 0 & \frac{c_\theta ^2}{2} \\
\end{array}
\right)
\end{equation}
The fixed state of this map is given by $\rho_{PC} ^* = \frac{1}{2}(\mathbf{1}+b_3\sigma_Z)$. This type of map is used in the study of thermalization of a single qubit and is called the homogenizer map \cite{ziman2001quantumhomogenization}. In certain limits, the map can become noise maps for amplitude damping and dephasing.\\
The correlations $c_{12} $ and $c_{21}$ are not a resource for this interaction since they can be generated under the resource theory where phase-covariant maps are the allowed free operations. Consequently, there presence does not affect the structure of the map, however their presence can affect the complete positivity of the map and the domain of positivity.\\
Since the evolution of the z-component--the population--is decoupled with the evolution of the x-y components--coherences--, these operations can never generate coherence between the $|0\rangle$ and $|1\rangle$ state of the system qubit for a system starting with a z-Bloch vector. The allowed operations restrict the class of reachable states. However, if system is allowed to interact with an environment/ancillary qubit that has coherences, i.e. $b_1, b_2 \neq 0$, e.g. $\rho_E = |+\rangle\langle+|,$ then under the same allowed operations, the system qubit can generate coherence between the $|0\rangle$ and $|1\rangle$ state of the system. In this way, a $|+\rangle$ is a resource for the theory since it expands the reachable set of states/transformations. Such a map would look like\\
\begin{equation}
\bar{\bar{\Phi}}^{\textcolor{black}{ \bullet}}=\left(
\begin{array}{cccc}
 1 & 0 & 0 & 0 \\
  0 & \frac{1}{2} c_\theta c_{h_+} & -\frac{1}{2} c_\theta s_{h_+} & \frac{1}{2} s_{\theta} (b_{1} s_{h_+}+b_{2} c_{h_+}) \\
 0 & \frac{1}{2} c_\theta s_{h_+} & \frac{1}{2} c_\theta c_{h_+} & \frac{1}{2} s_{\theta} (b_{2} s_{h_+}-b_{1} c_{h_+}) \\
 \frac{1}{2} s_{\theta}^2 b_{3}  & -\frac{1}{4} b_{2} s_{2\theta} & \frac{1}{4} b_{1} s_{2\theta} & \frac{c_\theta^2}{2} \\
\end{array}
\right)
\end{equation}
While this map allows for coherence generation, this map is no longer a thermal-operation nor is it Gibbs-preserving. This is because the action of the map on the Gibbs state would not be identity and would generate coherences when acted upon a Bloch vector $[0,0,r_G]^T$. The map fails to be a thermal operation because the environment now has coherence and is not in a Gibbs state. The map is however, still a covariant operation for the full system-environment composite.\\

\subsection{Thermal operations with system environment correlations}
We now consider the case where the following holds: (i) $[H_{SE},H_S+H_E]\implies h_1=h_2=h\implies h_- = 0,  h_+ = 2h \text{ and } \theta=J$ and (ii) $[H_E,\rho_E] = 0 \implies b_1 = b_2 = 0$. However, this time let the system and environment start with all possible two-body correlations. This implies introduction of resources. Correlations like $c_{32}$ and $c_{31}$ are resources because they cannot be generated freely by the allowed operations. The resulting map is given by,
\begin{equation}
    \bar{\bar{\Phi}}^{\textcolor{black}{ \bullet}} = \left(
\begin{array}{cccc}
 1 & 0 & 0 & 0 \\
 \frac{1}{2} s_\theta (c_{31} s_{h_+}+c_{32} c_{h_+}) & \frac{1}{2} c_{\theta} c_{h_+} & -\frac{1}{2}c_{\theta} s_{h_+} & 0 \\
 \frac{1}{2} s_\theta (c_{32} s_{h_+}-c_{31} c_{h_+}) & \frac{1}{2} c_{\theta} s_{h_+} & \frac{1}{2} c_{\theta} c_{h_+} & 0 \\
 \frac{1}{2} s_\theta (b_{3} s_\theta+2i\mathcal{I}(c_{21}) c_{\theta}) & 0 & 0 & \frac{c_{\theta} ^2}{2} \\
\end{array}
\right)
\end{equation}
Therefore, system-environment correlations are a resource for the system since they allow for asymmetric, non-zero shift in the x-y direction. The presence of the shift breaks the isotropic transformation of the x-y Bloch vector under a thermal operation. The fixed state associated with this map has non-zero Bloch vector components along all the three directions implying that the state has non-zero coherences.\\
However, this map is not a Gibbs-Preserving since the non-zero shift is not canceled out by non zero mixing between the z-x and z-y Bloch vectors, resulting in a fixed point that is not a Gibbs-like state. Therefore, this map while energy preserving and starting with $\rho_E = \frac{1}{2}(\mathbf{1}+b_3\sigma_z)$ is not a thermal operation or a Gibbs-preserving map. \\
\subsection{Correlation as a resource for single-shot Gibbs operations}
We have seen above that coherence in the environment/ancillary qubit $|+\rangle$ or system-environment correlations alone are not enough to generate a Gibbs-preserving map. While they can generate maps that have coherences, they do not have the Gibbs state as a fixed point. This is because for a Gibbs-preserving map we need the shift component of the affine map to cancel out the effect of the $T_{i3}$ component of the rotation part of the affine map. \\
To go beyond thermal operations to Gibbs-preserving maps, we need a map that has non-zero terms in the rotation matrix $T$ and has non-zero shift vector terms $\vec{\tau}$. We consider a scenario where the evolution is energy preserving, however the system has access to ancillary with $b_1 , b_2\neq0$ state and non-zero correlations terms $c_{23}, c_{13} \neq0$. Furthermore, we demand that 
\begin{align}
    \tau_i &= - r_G T_{i3} \hspace{0.5em} i= 1,2\\
    \tau_3 &=  r_G (1-T_{33} )
\end{align}
Under these conditions, we get a reduced form of the dynamical map $\Phi$ in equation \ref{eqn:GenMap}. The constraints on $\tau_i$ end up constraining the correlations to take specific values. In particular, 
\begin{align*}
    c_{13}&=-b_1 r_G\\
    c_{23}&= -b_2 r_G\\
    \mathcal{I}(c_{12})&= \frac{1}{2}\left(-\frac{1 + \sin^2(\theta)}{\sin(\theta)\cos(\theta)} \cdot r_G + b_3\tan(\theta)\right)
\end{align*}
\begin{equation}
    \bar{\Phi}_{GP}^{\textcolor{cyan}{\bullet}}=\left(
\begin{array}{cccc}
 1 & 0 & 0 & 0 \\
 -\frac{1}{2} r_G s_\theta \Delta_1 & \frac{1}{2} c_\theta c_{h_+} & -\frac{1}{2} c_\theta s_{h_+} & \frac{1}{2} s_\theta \Delta_1 \\
 \frac{1}{2} r_G s_\theta \Delta_2 & \frac{1}{2} c_\theta s_{h_+} & \frac{1}{2} c_\theta c_{h_+} & \frac{1}{2} s_\theta \Delta_2 \\
 r_G-\frac{1}{2} r_G c_\theta^2 & -\frac{1}{2} b_{2} s_\theta c_\theta & \frac{1}{2} b_{1} s_\theta c_\theta & \frac{c_\theta^2}{2} 
\end{array}
\right)
\end{equation}
where $\Delta_1 = (b_{1} s_{h_+}+b_{2} c_{h_+})$ and $\Delta_2 = (b_{1} c_{h_+}-b_{2} s_{h_+})$
The action of this map on a general state can mix all the three Bloch vector components as well as provide shift. The map has the Gibbs-state as the fixed point. However, this map is no longer phase-covariant. \\
This construction gives a Gibbs-preserving map, but it depends on having access to a system-ancilla correlated state with fine-tuned two-body correlations. Such two-body correlations are operationally hard to implement \cite{Fauseweh2024}. Below we present a protocol to generate such maps with initially uncorrelated qubits.
This map is an example of a map that while being Gibbs-preserving, does not obey the conditions for thermal operation and is not energy-conserving either.
\subsection{Three-qubit system}
We will consider the three-qubit Hamiltonian discussed in the paper. A schematic of the setup is shown in Fig. \ref{fig:three_qubit_schematic}.
This Hamiltonian is energy-preserving and consequently is still in the class of covariant operations on the composite. Furthermore, we let the environment be in the state $\rho_E = \frac{1}{2}(\mathbf{1}+b_3 \sigma_Z)$ which satisfies the commutation condition with the $H_E$. Lastly, since we start with initially uncorrelated input system-environment-resource state, we are still in the reasonable limit of systems being initially separable. This will imply that the operation is still a thermal operation except, it is thermal operations with respect to $\rho_S\otimes\rho_R$. Due to the presence of the coherence term in the resource qubit and the partial swap-like interaction, the coherence from the resource will be used up by the input system to generate Gibbs-preserving dynamics with respect to a thermal state $\rho_T = \frac{1}{2}(\mathbf{1}+r_G \sigma_Z)$.\\
Below is the dynamical map for the dynamics. Here, we have redefined some variables for cleaner representation. The compact notations represent $J' = \sqrt{2} J$, $s_J = \sin(\sqrt{2}J)$, $c_J = \cos(\sqrt{2}J)$, $s_{2J} = \sin(2\sqrt{2}J)$, $c_{2J} = \cos(2\sqrt{2}J)$, $s_{4J} = \sin(4\sqrt{2}J)$, $c_{4J} = \cos(4\sqrt{2}J)$, $s_{2h} = \sin(2h)$, and $c_{2h} = \cos(2h)$. We further define $A = 4(b_3 f_3 + 1)c_{2J} - (b_3 f_3 - 1)(c_{4J} + 3)$, $\Omega_1 = f_1 s_{2h} + f_2 c_{2h}$, $\Omega_2 = f_2 s_{2h} - f_1 c_{2h}$, $B_1 = \sqrt{2}b_3 s_J^3 c_J$, $B_2 = \sqrt{2}s_J c_J^3$, $\phi_+ = \frac{1}{8}\cos{2 h}$, and $\phi_- =\frac{1}{8}\sin{2 h}$.
\begin{equation}
\Phi_{GP}^{\textcolor{cyan}{ \bullet}} =\begin{pmatrix}
1 & 0 & 0 & 0 \\
B_1\Omega_1 & A\phi_+ & A\phi_- & B_2\Omega_1 \\
B_1\Omega_2 & A\phi_- & A\phi_+ & B_2\Omega_2 \\
\frac{(b_3+f_3)s_{2J}^2}{2} & -\frac{f_2 s_{4J}}{2\sqrt{2}} & \frac{f_1 s_{4J}}{2\sqrt{2}} & c_{2J}^2
\end{pmatrix}
\end{equation}
In we let $f_1 = f_2 = 0$ this map becomes phase-covariant and if we let $f_1, f_2 \neq 0$ but unconstrained, it gives a globally energy preserving map $\Phi_E$. For the map to be Gibbs-preserving with respect to a thermal state with z-Bloch vector, $r_G$, the following must be satisfied,
\begin{align}
    B_1 = -r_G B_2 &\implies b_3 s_{J}^2= -r_Gc_J ^2 \\
    b_3 +f_3 &= 2r_G 
\end{align}
This means that for a given $b_3$, the interaction parameter $J = \frac{1}{\sqrt{2}}\arctan[\sqrt{\frac{-r_G}{b_3}}]$. For this to be a real solution, the allowed values of $r_G$ are ones that have z-Bloch vector in opposite direction of $b_3$. Once this is satisfied, the resource state parameters $f_1$ and $f_2$ can be arbitrary but non-zero. The parameter $f_3$ is constrained to be $f_3 = -b_3 (2\tan^2(J) +1)$. Along with these a global constraint of $f_i$ is that $|f|^2\leq1$. These constrain the allowed operations however provide a subset of Gibbs-preserving maps to be accessible.\\
Therefore, given access to a thermal operation at temperature $T'$ corresponding to the environment qubit z-Bloch vector $b_3$, we can get a Gibbs-preserving map that preserves a Gibbs-state with z-Bloch vector $r_G$ by tuning J to be such that $ r_G = -b_3 \tan^2(\sqrt{2}J)$ and introducing a resource qubit with $f_3 = -b_3 (1+2\tan^2(\sqrt{2}J))$. When $f_1 = f_2 = 0$, the map is trivially Gibbs-preserving since it falls back into the class of phase-covariant maps. Therefore, the coherence terms will acts as a switch to go between $\Phi_{PC}$ and $\Phi_{GP}$.\\
The eigenvector of this dynamical map with the eigenvalue 1 correpsonds to the fixed state of the map. The unconstrained map, with only global energy preservation, has the following fixed state
\begin{equation}
\begin{pmatrix}
    1 \\
    -\frac{\sin \left(2 \sqrt{2} J T\right) (2 B_1+B_2 (b_3+f_3)) (-A \Omega_1 \phi_++A \Omega_2 \phi_-+\Omega_1)}{\sqrt{2} B_2 \cos \left(2 \sqrt{2} J T\right) (f_1 (A \Omega_1 \phi_--A \Omega_2 \phi_++\Omega_2)-f_2 (A \Omega_2 \phi_-+\Omega_1)+A f_2 \Omega_1 \phi_+)+2 (A (\phi_--\phi_+)+1) (A (\phi_-+\phi_+)-1) \sin \left(2 \sqrt{2} J T\right)}\\
    \frac{\sin \left(2 \sqrt{2} J T\right) (2 B_1+B_2 (b_3+f_3)) (A \Omega_1 \phi_--A \Omega_2 \phi_++\Omega_2)}{\sqrt{2} B_2 \cos \left(2 \sqrt{2} J T\right) (-f_1 (A \Omega_1 \phi_-+\Omega_2)+A f_1 \Omega_2 \phi_++f_2 (-A \Omega_1 \phi_++A \Omega_2 \phi_-+\Omega_1))-2 (A (\phi_--\phi_+)+1) (A (\phi_-+\phi_+)-1) \sin \left(2 \sqrt{2} J T\right)}\\
    -\frac{\sqrt{2} B_1 \cos \left(2 \sqrt{2} J T\right) (f_1 (A \Omega_1 \phi_--A \Omega_2 \phi_++\Omega_2)-f_2 (A \Omega_2 \phi_-+\Omega_1)+A f_2 \Omega_1 \phi_+)+(b_3+f_3) (A (\phi_+-\phi_-)-1) (A (\phi_-+\phi_+)-1) \sin \left(2 \sqrt{2} J T\right)}{\sqrt{2} B_2 \cos \left(2 \sqrt{2} J T\right) (f_1 (A \Omega_1 \phi_--A \Omega_2 \phi_++\Omega_2)-f_2 (A \Omega_2 \phi_-+\Omega_1)+A f_2 \Omega_1 \phi_+)+2 (A (\phi_--\phi_+)+1) (A (\phi_-+\phi_+)-1) \sin \left(2 \sqrt{2} J T\right)}
\end{pmatrix}
\end{equation}
under the condition $f_1 = f_2 = 0$, the fixed state of the map is $\vec{v} = [1,0,0,\frac{b_3+f_3}{2}]$. 

\end{widetext}
\end{document}